\documentclass[journal,draftclsnofoot,onecolumn,10pt]{IEEEtran}
\usepackage{graphicx}
\usepackage{multirow}
\usepackage{subfigure}
\usepackage{cite}
\usepackage{booktabs}
\usepackage{threeparttable}
\usepackage{graphicx}
\usepackage{url}
\usepackage[cmex10]{amsmath}
\usepackage{amssymb}
\usepackage{esint}
\usepackage{cool}
\usepackage{multirow}
\usepackage{hhline}
\usepackage{subeqnarray}

\newcommand{\RP}{r_P}
\newcommand{\RS}{r_S}
\newcommand{\RK}{r_K}

\newcommand{\mean}{\mathbb{E}}
\newcommand{\cov}{\mathbb{C}}
\newcommand{\var}{\mathbb{V}}

\newcommand{\define}{\triangleq}

\newcommand{\bias}{\mathrm{BIAS}}
\newcommand{\mse}{\mathrm{MSE}}
\usepackage{tikz,pgfplots}
\usetikzlibrary{arrows}

\newtheorem{Lem}{Lemma}
\newtheorem{Thm}{Theorem}

\newtheorem{Rmk}{Remark}

\newtheorem{Cly}{Corollary}

\newcommand{\corr}{\mathrm{corr}}
\newcommand{\ARE}{\mathrm{ARE}}
\newcommand{\sgn}{\mathrm{sgn}}

\begin{document}

\title{Comparison of Spearman's rho and Kendall's tau in Normal and Contaminated Normal Models}
\author{$\text{Weichao~Xu}^*$,~\IEEEmembership{Member,~IEEE,} Yunhe Hou,~\IEEEmembership{Member,~IEEE,} Y.~S.~Hung,~\IEEEmembership{Senior~Member,~IEEE},\\ and Yuexian Zou,~\IEEEmembership{Senior~Member,~IEEE}%
\thanks{This work was supported in part by the University of Hong Kong under Small Project Grant 200807176233 and Seed Funding Programme for 
Basic Research 201001159007.}%
\thanks{ W.~Xu, Y.~Hou and Y.~S.~Hung are with the Dept. of Electrical and Electronic
 Engineering, The University of Hong Kong, Pokfulam Road, Hong Kong, Hong Kong (e-mail:wcxu@eee.hku.hk; yhhou@eee.hku.hk; yshung@eee.hku.hk).}%
\thanks{ Y.~Zou is with the Advanced Digital Signal Processing Lab, Peking University Shenzhen Graduate School, Shenzhen, Guangdong 518055, P. R. China (e-mail:zouyx@szpku.edu.cn).}%
\thanks{${}^*$Corresponding Author. }
\thanks{Tel:+852-28578489 Fax:+852-25598738 (W. Xu, Y. Hou and Y. S. Hung). }
\thanks{Tel:+86-755-26032016 Fax:+86-755-26032016 (Y.~Zou). }%
}

\markboth{Manuscript submitted to IEEE Transactions on Information Theory}{Xu \MakeLowercase{\textit{et al.}}: Spearman's Rho or Kendall's Tau}

%
%
\maketitle


\begin{abstract}
This paper analyzes the performances of the Spearman's rho (SR) and Kendall's tau (KT)   
with respect to samples drawn from bivariate normal and bivariate contaminated normal populations.
The exact analytical formulae of the variance of SR and the covariance between SR and KT are obtained 
based on the Childs's reduction formula for the quadrivariate normal positive orthant probabilities. 
Close form expressions with respect to the expectations of SR and KT are established 
under the bivariate contaminated normal models.  The bias, mean square error (MSE) and asymptotic relative
efficiency (ARE) of the three estimators based on SR and KT to the Pearson's product moment
correlation coefficient (PPMCC) are investigated in both the normal and contaminated normal models. 
Theoretical and simulation results suggest that, contrary to the opinion 
of equivalence between SR and KT in some literature,
the behaviors of SR and KT are strikingly different in the aspects of bias effect, variance, mean square error, 
and asymptotic relative efficiency. 
The new findings revealed in this work provide not only deeper insights into the
two most widely used rank based correlation coefficients,
but also a guidance for choosing which one to use under the circumstances where the PPMCC fails to apply. 
\end{abstract}
\begin{keywords}
Bivariate normal, 
Correlation theory,
Contaminated normal model,
Kedall's tau (KT),
Orthant probability,
Pearson's product moment correlation coefficient (PPMCC),
Quadrivariate normal,
Spearman's rho (SR).
\end{keywords}
\section{Introduction} \label{sect:introduction}
\PARstart{C}{orrelation} analysis is  among the  core  research  paradigms  in  nearly  all  branches 
of scientific and engineering fields,  not to  mention the area of information theory%
~\cite{ruchkin1965,mccheng1968,chadwick1969,hansen1970,goldstein1973,bershad1974,Bae2006,johansen2008,johansen2009,gomadam2009}.
Being interpreted as the strength of statistical relationship between two random variables~\cite{kendall91}, 
correlation should be large and positive if  there is a high probability that
large (small) values of one variable occur in conjunction with large (small) values of another; 
and it should be large and negative if the direction is reversed~\cite{non_JDG92}. 
A number of methods have been proposed and applied in the literature 
to assess the correlation between two random variables.
Among these methods the Pearson's product moment correlation coefficient (PPMCC)~\cite{fisher21,fisher90},
Spearman's rho (SR)~\cite{kendall90} and Kendall's tau (KT)~\cite{kendall90} are perhaps the most
widely used~\cite{mk01}. 

The properties of PPMCC in bivariate normal samples (binormal model) is well known thanks to 
the creative work of Fisher~\cite{fisher21}. 
It follows that, in the normal cases, 
1) PPMCC is an asymptotic unbiased estimator of the population correlation $\rho$, and
2) the variance of PPMCC approaches the Cramer-Rao lower bound (CRLB) with increase of 
the sample size~\cite{kendall91}. Due to its optimality, PPMCC 
has and will continue to play the dominant role when quantifying the 
intensity of correlation between bivariate random variables in the literature. 
However, sometimes the PPMCC might not be applicable when the following scenarios happen: 
\begin{enumerate}
	\item The data is incomplete, that is, only ordinal information (e.g. ranks) is available.
		This situation is not uncommon in the area of social sciences, such as 
		psychology and education~\cite{kendall90};
	\item The underlying data is complete (cardinal) and follows a bivariate normal distribution, but is attenuated more or less by some 
		monotone nonlinearity in the transfer characteristics of sensors~\cite{tumanski2006};
		\item The data is complete and the majority follows a bivariate normal distribution, but there exists a tiny fraction of outliers (impulsive noise)~\cite{stein1995,wang2000it,reznic2002}. \label{enu:outlier}
\end{enumerate}
Under these circumstances, it would be more suitable to employ the two most popular nonparametric coefficients, SR and KT,
which are 1) dependant only on ranks, 2) invariant under increasing monotone transformations~\cite{kendall90},
and 3) robust against impulsive noise~\cite{georgy2002}.
Now we are at a stage to ask the question: which one, SR or KT, should we use in Scenarios 1) to 3)
where the familiar PPMCC is inapplicable? Unfortunately, however,
despite the rich history of SR and KT, the answers to this question are still inconsistent in the literature. 
Some researchers, such as Fieller \emph{et al}\cite{fhp1957}, preferred KT to SR based on empirical evidences; while others,
such as Gilpin~\cite{gilpin1993}, asserted that SR and KT are equivalent.  

Aiming at resolving such inconsistency, in this work we investigate systematically the properties of SR and KT
under the binormal model~\cite{lapidoth2010a,lapidoth2010b,bross2010}.  Moreover, to deal with Scenario~\ref{enu:outlier}) mentioned above, we also investigate their properties under the contaminated binormal model~\cite{stein1995,wang2000it,reznic2002}.
Our theoretical contribution is multifold. 
Firstly, we find a computationally more tractable formula of the variance of SR. Based on this formula, we
provide the densely tabulated Table~\ref{tab:tabO1O2} with high precision (ten decimal places). 
This table overcomes the shortcomings of the existing power-series-based approximations that are tedious to use and of rather
limited precision (up to five decimal places and for $\rho\leq0.8$ only)~\cite{kendall1949,dks1951,fhp1957,david1961}. 
Secondly, we derive the \emph{exact} analytical expression of the covariance between SR and KT.
With this new analytical result, we uncover a minor error in the literature~\cite{dks1951,kendall90}. 
Thirdly, we obtain the asymptotic expressions of the variances and hence the  asymptotic relative
efficiencies (AREs) concerning the three estimators of the population correlation $\rho$ based on SR and KT.
Finally, we find the asymptotic expressions with respect to the expectations of SR and KT under the contaminated normal model. 


The rest part of this paper is structured as follows. Section~\ref{sect:sectgeneral} 
gives some basic definitions and summarizes the general properties of PPMCC, SR and KT. In Section~\ref{sect:sectlemmas},
we lay the foundation of the theoretical framework in this study by outlining some critical results in the binormal model.  
Section~\ref{sect:sectmain}  establishes, in the bivariate normal model, 
1) the exact expression of the variance of SR,  
2) two exact expressions concerning the covariance between SR and KT, and, 
3) in the contaminated normal model, the closed form formulae associated with the expectations of SR and KT, respectively.  
In Section~\ref{sect:sectestimators} we focus on the performances of the three estimators of $\rho$ constructed from 
SR and KT.  Section~\ref{sect:sectnum} verifies the analytical results with Monte Carlo simulations.
Finally, in Section~\ref{sect:sectconclusion} we provide our answers to the
above raised question  concerning the choice of Spearman's rho and Kendall's tau in practice when PPMCC fails to apply. 
\section{Basic Definitions and General Properties} \label{sect:sectgeneral} 
\subsection{Definitions}
Let $\{(X_i ,Y_i)\}_{i=1}^n$ denote $n$ independent and identically distributed
(i.i.d.) data pairs drawn from a bivariate population with continuous joint
distribution.  Suppose that $X_j$ is at the $k$th position in the sorted sequence  $X_{(1)}<\cdots< X_{(n)}$. 
The number $k$ is termed the \emph{rank} of $X_j$ and is denoted by $P_j$. Similarly we can get the rank of $Y_j$ which is denoted by $Q_j$~\cite{kendall90}.
Let $\bar{X}$ and $\bar{Y}$ be the arithmetic mean values of $X_i$ and $Y_i$, respectively. Let $\sgn(\blacktriangle)$ stand for the sign of the argument $\blacktriangle$.  
The three well known classical correlation coefficient, PPMCC ($r_P$), SR ($r_S$), and KT ($r_K$), are then defined as follows~\cite{non_JDG92}:
\begin{align}
	\label{eq:rpdef}
	\RP(X,Y)&\define\frac{\sum\limits_{i=1}^n \left(X_i-\bar X\right)\left(Y_i-\bar Y\right)}{\left[\sum\limits_{i=1}^n\left(X_i-\bar X\right)^2\sum\limits_{i=1}^n\left(Y_i-\bar Y\right)^2\right]^\frac{1}{2}}\\
	\label{eq:rsdef}
  \RS(X,Y)&\define 1-\frac{6\sum\limits_{i=1}^n (P_i -Q_i )^2}{n(n^2-1)}\\
	\label{eq:rkdef}
	\RK(X,Y)&\define \frac{\underset{i\ne j=1}{\sum\limits^n\sum\limits^n}\,\,\sgn\left(X_i-X_j\right) \sgn\left(Y_i-Y_j\right)}{n(n-1)}. 
\end{align}
To ease the following discussion, we will employ the symbol $r_\lambda(X,Y)$, $\lambda\in\{P,S,K\}$ as a compact notation for
the three coefficients. For brevity, the arguments of $r_\lambda(X,Y)$ will be dropped in the sequel unless ambiguity occurs.
\subsection{General Properties}
It follows that coefficients $r_\lambda$, $\lambda\in\{P,S,K\}$ possess the following general properties:
\begin{enumerate}
\item $r_\lambda(X,Y) \in [-1,1]$ for all $(X,Y)$ (standardization);
\item $r_\lambda(X,Y)=r_\lambda(Y,X)$ (symmetry);
\item $r_\lambda=\pm1$ if $Y$ is a positive (negative) linear transformation of $X$ (shift and scale invariance);
\item $\RS{=}\RK{=}\pm1$ if $Y$ is a monotone increasing (decreasing) function of $X$ (monotone invariance);
\item The expectations of $r_\lambda$ equal zero if $X$ and $Y$ are independent (independence);
\item $r_\lambda(+,+)=-r_\lambda(-,+)=-r_\lambda(+,-)=r_\lambda(-,-)$;
\item $r_\lambda$ converges to normal distribution when the sample size $n$ is large.
\end{enumerate}
Note that the first six properties are discussed in~\cite{non_JDG92} and~\cite{mk01}, and the last property follows from 
the asymptotic theory of $U$-statistics established by Hoeffding~\cite{hoeffding1948}.
\subsection{Relationships Among PPMCC, SR and KT}
From their expressions (\ref{eq:rpdef})--(\ref{eq:rkdef}), it appears that the three coefficients PPMCC, SR and KT are quite different. 
However, as demonstrated below, these three coefficients are closely related with each other.
\subsubsection{Daniel's Generalized Correlation Coefficient}
Consider the $n$ data pairs $(X_i,Y_i)$, $i=1,\ldots,n$, at hand. To each pair of $X$'s, ($X_i, X_j$), 
we can allot a score $a_{ij}$ such that $a_{ij}=-a_{ji}$ and $a_{ii}=0$. 
In a similar manner, we can also allot a sore $b_{ij}$ to the ordered pair of $Y$'s, ($Y_i, Y_j$).
The Daniel's generalized coefficient $\Gamma$ is then defined by~\cite{daniels1944}
\begin{equation}
	\Gamma \define \frac{\sum\limits_{i=1}^n\sum\limits_{j=1}^n a_{ij} b_{ij}}{\left(\sum\limits_{i=1}^n\sum\limits_{j=1}^n a_{ij}^2 \sum\limits_{i=1}^n\sum\limits_{j=1}^n b_{ij}^2\right)^{\frac{1}{2}}}.
	\label{eq:gammadef}
\end{equation}
This general setup covers PPMCC, SR and KT as special cases with respect to different systems of scores~\cite{daniels1944}: 
\begin{itemize}
	\item  Replacing $a_{ij}$ by $X_j{-}X_i$ and $b_{ij}$ by $Y_j{-}Y_i$ in (\ref{eq:gammadef}) gives the PPMCC $r_P$ defined in (\ref{eq:rpdef});
	\item  Replacing $a_{ij}$ by $P_j{-}P_i$ and $b_{ij}$ by $Q_j{-}Q_i$ in (\ref{eq:gammadef}) gives the SR $r_S$ defined in (\ref{eq:rsdef});
	\item  Replacing $a_{ij}$ by $\sgn(X_j{-}X_i)$ and $b_{ij}$ by $\sgn(Y_j{-}Y_i)$ in (\ref{eq:gammadef}) gives the KT $r_K$ defined in (\ref{eq:rkdef}).
\end{itemize}
\subsubsection{Inequalities between SR and KT}
It is possible to state certain inequalities connecting the values of SR and KT based on a given set of $n$ observations.
The first one, ascribed to Daniel~\cite{daniels1950}, is 
\begin{equation}
	-1\leq \frac{3(n+2)}{n-2}\RK-\frac{2(n+1)}{n-2}\RS \leq 1	
	\label{eq:daniel}
\end{equation}
which, for large $n$, becomes 
\begin{equation*}
	-1 \leq 3\RK-2\RS \leq 1.
\end{equation*}
The second one, due to Durbin and Stuat~\cite{durbinstuart1951}, states that
\begin{equation}
	\RS \leq 1-\frac{1-\RK}{2(n+1)}\left[(n-1)(1-\RK)+4\right].
	\label{eq:durbin}
\end{equation}
Combing (\ref{eq:daniel}) and (\ref{eq:durbin}) and letting $n\to\infty$ yield the bounds of SR, in terms of KT, as
\begin{align*}
	\frac{3}{2}\RK-\frac{1}{2} &\leq\RS\leq \frac{1}{2} + \RK -\frac{1}{2}\RK^2, \quad \RK \geq 0 \\
	\frac{3}{2}\RK+\frac{1}{2} &\geq\RS\geq \frac{1}{2}\RK^2 + \RK -\frac{1}{2}, \quad \RK \leq 0. \\
\end{align*}
\subsubsection{Relationship of SR to Other Coefficients} Besides the PPMCC and KT, SR is also closely related to other 
correlation coefficients, e.g., the \emph{order statistics correlation coefficient} (OSCC)~\cite{xu2006,xu2007,xu2008}
and the \emph{Gini correlation} (GC)~\cite{gini87}. In fact, SR can be reduced from  
the OSCC and GC by replacing the variates with corresponding ranks~\cite{xu2010}.

\section{Auxiliary Results in Normal Cases} \label{sect:sectlemmas}
In this section we provide some prerequisites concerning the orthant probabilities of normal distributions. These probabilities, contained in
Lemma~\ref{lem:orthant}, are critical for the development of Theorem~\ref{thm:varrs} and Theorem~\ref{thm:covrsrk} later on.
Moreover, some well known results about the expectation and variance of PPMCC, SR and KT are collected in Lemma~\ref{lem:lemknown}
for ease of exposition.  For convenience, we use symbols $\mean(\blacktriangle)$, $\var(\blacktriangle)$,
$\cov(\blacktriangle,\blacklozenge)$, and $\corr(\blacktriangle,\blacklozenge)$ 
in the sequel to denote the mean, variance, covariance, and correlation of (between) random variables, respectively. 
Symbols of \emph{big oh} and \emph{little oh} are utilized to compare the magnitudes of two
functions $u(\blacktriangle)$ and $v(\blacktriangle)$ as the argument $\blacktriangle$ tends to a limit $L$ (might be infinite).
The notation $u(\blacktriangle)=O(v(\blacktriangle))$, $\blacktriangle {\to} L$, denotes that
$|u(\blacktriangle)/v(\blacktriangle)|$ remains bounded as $\blacktriangle {\to} L$; whereas the 
notation $u(\blacktriangle)=o(v(\blacktriangle))$, $\blacktriangle {\to} L$, denotes
that $u(\blacktriangle)/v(\blacktriangle){\to} 0$ as $\blacktriangle {\to} L$~\cite{Serfling2002}.  
Symbols of $P_m^0(Z_1,\ldots,Z_m)$ are adopted to denote the positive orthant probabilities 
associated with multivariate normal random vectors $[Z_1 \cdots Z_m]$  of dimensions $m=1,\ldots,4$, respectively. 
The notation $R(\varrho_{rs})_{m\times m}$ stands for correlation matrix with each element 
$\varrho_{rs}\define \corr(Z_r,Z_s)$,  $r,s=1,\ldots,m$. Obviously the diagonal entries in $R$ are all unities.
For compactness, we will also use the symbol $P_m^0(R)$ to denote $P_m^0(Z_1,\ldots,Z_m)$ in the sequel. 
\subsection{Orthant Probabilities for Normal Distributions}
\begin{Lem}
	\label{lem:orthant}
Assume that $Z_1$, $Z_2$, $Z_3$, $Z_4$ follow
a quadrivariate normal distribution with zero means and  correlation matrix  $R=\left(\varrho_{rs}\right)_{4\times 4}$. 
Define 
\begin{equation}
	\label{eq:Hdef}
	H(\blacktriangle)  \define 
	\begin{cases}
		1 &\quad (\blacktriangle > 0)\\
		0 &\quad (\blacktriangle \leq 0).
	\end{cases}
\end{equation}
Then the orthant probabilities
\begin{align}
	\label{eq:p1}
	P_1^0(Z_1) &\define \mean\left\{H(Z_1)\right\} \notag\\
	&=\frac{1}{2}\\
	P_2^0(Z_1,Z_2) &\define \mean\left\{H(Z_1)H(Z_2)\right\} \notag \\
	\label{eq:p2}
	&=\frac{1}{4}\left(1+\frac{2}{\pi}\sin^{-1}\varrho_{12}\right)\\
	P_3^0(Z_1,Z_2,Z_3) &\define \mean\left\{H(Z_1)H(Z_2)H(Z_3)\right\} \notag \\
	\label{eq:p3}
	&=\frac{1}{8}\left(1+\frac{2}{\pi}\sum_{r=1}^2\sum_{s=r+1}^3 \sin^{-1}\varrho_{rs}\right)\\
	P_4^0(Z_1,Z_2,Z_3,Z_4) &\define \mean\left\{H(Z_1)H(Z_2)H(Z_3)H(Z_4)\right\}\notag \\
	\label{eq:p4}
	&=\frac{1}{16}\left(1{+}\frac{2}{\pi}\sum_{r=1}^3\sum_{s=r+1}^4 \sin^{-1}\varrho_{rs}{+}W\right)
\end{align}
where 
\begin{align}
	\label{eq:lemWdef}
	W &{\define}\frac{1}{\pi^4}\overset{+\infty}{\iiiint\limits_{-\infty}}\frac{\exp\left(-\frac{1}{2} z R z^T \right)}{z_1 z_2 z_3 z_4} dz_1 dz_2 dz_3 dz_4\\
	\label{eq:lemWreduct}
	&{=}\sum_{\ell=2}^4\frac{4}{\pi^2}\int_0^1\hspace{-3pt}\frac{\varrho_{1\ell}}{\left[1{-}\varrho_{1\ell}^2 u^2\right]^{\frac{1}{2}}}\sin^{-1}\left[\frac{\alpha_{\ell}(u)}{\beta_{\ell}(u)\gamma_{\ell}(u)}\right] du 
\end{align}
with
\begin{align*}
	\alpha_2 & {=} \varrho_{34}{-}\varrho_{23}\varrho_{24}{-}[\varrho_{13}\varrho_{14}{+}\varrho_{12}(\varrho_{12}\varrho_{34}{-}\varrho_{14}\varrho_{23}{-}\varrho_{13}\varrho_{24})]u^2\\
	\alpha_3 & {=} \varrho_{24}{-}\varrho_{23}\varrho_{34}{-}[\varrho_{12}\varrho_{14}{+}\varrho_{13}(\varrho_{13}\varrho_{24}{-}\varrho_{14}\varrho_{23}{-}\varrho_{12}\varrho_{34})]u^2\\
	\alpha_4 & {=} \varrho_{23}{-}\varrho_{24}\varrho_{34}{-}[\varrho_{12}\varrho_{13}{+}\varrho_{14}(\varrho_{14}\varrho_{23}{-}\varrho_{13}\varrho_{24}{-}\varrho_{12}\varrho_{34})]u^2\\
	\beta_2  & {=}\beta_3{=}\left[1{-}\varrho_{23}^2{-}(\varrho_{12}^2{+}\varrho_{13}^2{-}2\varrho_{12}\varrho_{13}\varrho_{23})u^2\right]^{\frac{1}{2}}\\
	\gamma_2 & {=} \beta_4{=}\left[1{-}\varrho_{24}^2{-}(\varrho_{12}^2{+}\varrho_{14}^2{-}2\varrho_{12}\varrho_{14}\varrho_{24})u^2\right]^{\frac{1}{2}}\\
	\gamma_3 & {=}\gamma_4{=} \left[1{-}\varrho_{34}^2{-}(\varrho_{13}^2{+}\varrho_{14}^2{-}2\varrho_{13}\varrho_{14}\varrho_{34})u^2\right]^{\frac{1}{2}}.
\end{align*}
\end{Lem} 
\begin{proof}
	The first statement (\ref{eq:p1}) is trivial. The second one (\ref{eq:p2}) is usually called \emph{Sheppard's theorem}
	in the literature, although it was proposed earlier by Stieltjes~\cite{kendall94}. The third one (\ref{eq:p3}) is a simple 
	generalization of Sheppard's theorem based on the relationship~\cite{gupta1963}
	\[
	P_3^0 = \frac{1}{2}\left[1-\sum_{r=1}^3 P_1^0(Z_r) +\sum_{r=1}^2\sum_{s=r+1}^3 P_2^0(Z_r,Z_s)\right].
	\]
	The last one (\ref{eq:p4}) is due to Childs~\cite{childs1967} and is termed the \emph{Childs's reduction formula}
	throughout. 
\end{proof}

\subsection{Some Well Known Results}
\begin{Lem} \label{lem:lemknown}
Let $\{(X_i ,Y_i)\}_{i=1}^n$ denote $n$ i.i.d. bivariate normal data pairs with correlation coefficient $\rho$.
Let $\RP$, $\RS$ and $\RK$ be the PPMCC, SR and KT that defined in (\ref{eq:rpdef})--(\ref{eq:rkdef}), respectively.
Write $S_1\define \sin^{-1}\rho$ and $S_2\define \sin^{-1}\frac{1}{2}\rho$.
Then
\begin{align}
	\label{eq:lemmeanrp}
	\mean(\RP) &= \rho\left[1-\frac{1-\rho^2}{2n}+O\left(n^{-2}\right)\right]\to\rho\text{ as }n\to\infty\\
	\label{eq:lemvarrp}
	\var(\RP)  &= \frac{(1-\rho^2)^2}{n-1}+O\left(n^{-2}\right)\\
	\label{eq:lemmeanrs}
	\mean(\RS) &= \frac{6}{\pi(n+1)}\left[\sin^{-1}\rho+(n-2)\sin^{-1}\frac{\rho}{2}\right]\\
	\label{eq:lemmeanrsasymp}
	&\to\frac{6}{\pi}\sin^{-1}\frac{\rho}{2} \text{ as }n\to\infty\\
	\label{eq:lemmeanrk}
	\mean(\RK) &= \frac{2}{\pi}\sin^{-1}\rho\\
	\label{eq:lemvarrk}
	\var(\RK)  &= \frac{2}{n(n-1)}\left[1{-}\frac{4S_1^2}{\pi^2}{+}2(n{-}2)\left(\frac{1}{9}{-}\frac{4S_2^2}{\pi^2}\right)\right].
\end{align}
\end{Lem}
\begin{proof}
	The first three results, (\ref{eq:lemmeanrp})--(\ref{eq:lemmeanrs}), were given by Hotelling~\cite{hotelling1953},
	Fisher~\cite{fisher90}, and Moran~\cite{moran1948}, respectively; whereas the last two results,
	(\ref{eq:lemmeanrk}) and (\ref{eq:lemvarrk}), were derived by Esscher \cite{esscher24}. 
\end{proof}
\section{Main Results in Normal and Contaminated Normal Models} \label{sect:sectmain}
In this section we establish our main results concerning $\var(\RS)$ and $\cov(\RS,\RK)$ in
the normal model as well as $\mean(\RS)$ and $\mean(\RK)$ in the  contaminated normal model.
We start from revisiting $\var(\RS)$ in normal samples. Being the most challenging part and of fundamental importance
for further development, the new discovery on  $\var(\RS)$ deserves to be formulated as a theorem. 
\subsection{Variance of Spearman's rho} \label{sect:sectvarrs}
\begin{Thm} \label{thm:varrs}
	Let $\{(X_i ,Y_i)\}_{i=1}^n$, $S_1$ and $S_2$ be defined as in Lemma~\ref{lem:lemknown}. 
	Write $\xi\in\{c,d,e,f,g,h,l,m,n,o,p,q\}$.
	Let $W_\xi$ be defined as in (\ref{eq:lemWdef}) with respect to $R_\xi$ that tabulated in Table~\ref{tab:orthant}.
	Then the variance of $\RS(X,Y)$ is 
\begin{equation}
	\begin{split}
		&\hspace{-6pt}\var(\RS){=}\frac{6}{n(n{+}1)}{+}\frac{9(n{-}2)(n{-}3)}{n(n^2{-}1)(n{+}1)}\bigg[(n{-}4)\Omega_1(\rho){+}\Omega_2(\rho)\bigg]\\
	&\quad{-}\frac{36}{\pi^2n(n^2{-}1)(n{+}1)}\bigg[3(n{-}2)(3n^2{-}15n{+}22)S_2^2\\
	&\hspace{4cm}{+}12(n{-}2)^2S_1S_2{-}2(n{-}3)S_1^2\bigg]
\end{split}
\label{eq:thmvarrs}
\end{equation}
where
\begin{align}
	\label{eq:thmO1}
	\Omega_1(\rho) 
	    & = W_c+8W_d+2W_f \\
	\label{eq:thmO2}
	   \Omega_2(\rho)
	&= 6W_g+8W_h+6W_l+2W_n+W_o+\frac{1}{3}.
\end{align}
Moreover, when $n$ is sufficiently large,
\begin{equation}
	\var(\RS)\simeq\frac{1}{n}\left[9\Omega_1(\rho)-\frac{324 S_2^2}{\pi^2}\right].
	\label{eq:thmvarrsasymp}
\end{equation}
\end{Thm}
\begin{proof}
	See Appendix~\ref{app:appvarrs}.
\end{proof}
\begin{Rmk}
	Unlike the Taylor-expansion-based approximate formulae in the literature~\cite{kendall1949,dks1951,fhp1957,david1961},
	the expression (\ref{eq:thmvarrs}) 
	in Theorem~\ref{thm:varrs} is \emph{exact} for both the sample size $n\ge 4$ and the population correlation $\rho\in[-1,1]$. 
	However, due to the complicated integrals involved in the expressions of $W$-terms in $\Omega_1(\rho)$ and $\Omega_2(\rho)$,
	the variance of $\RS$ cannot be expressed into elementary functions in general.  In other words,
	we need to conduct numerical integrations based on Childs's reduction formula (\ref{eq:lemWreduct}) so as to calculate $\Omega_1(\rho)$ and $\Omega_2(\rho)$ and
	hence $\var(\RS)$ from (\ref{eq:thmvarrs}). Nevertheless, exact results can be obtained for some particular cases.
	It can be shown that (Appendix~\ref{app:appO1O2O3}) 
	\begin{align}
		\label{eq:Omega120}
		\Omega_1(0)&=\frac{1}{9},\quad \Omega_2(0)=\frac{5}{9},\\
		\label{eq:Omega121}
		\Omega_1(1)&=1,\quad\,           \Omega_2(1)=\frac{16}{3}.
	\end{align}
	Substituting $\rho=0$ and (\ref{eq:Omega120}) into (\ref{eq:thmvarrs}) leads directly to
	\begin{equation}
		\label{eq:varrsrho0}
		\var(\RS)\big|_{\rho=0}=\frac{1}{n-1}	
	\end{equation}
	which is a well known result~\cite{kendall90}. Substituting $\rho=1$ and (\ref{eq:Omega121}) into (\ref{eq:thmvarrs}) 
	and (\ref{eq:thmvarrsasymp}) together with some simplifications yields
	\begin{equation}
		\label{eq:varrsrho1}
		\var(\RS)\big|_{\rho=1}=0	
	\end{equation}
	which is of no surprise but, to our knowledge, has never been proven explictly in the literature
	(although indirect arguments can be found~\cite{xu2010}). Note that $\var(\RS)$ also vanishes for $\rho=-1$ due to 
	symmetry.
\end{Rmk}
\subsection{Covariance between Spearman's rho and Kendall's tau} \label{sect:sectcovrsrk}
Besides the variance of SR just established in Theorem~\ref{thm:varrs}, the covariance between SR and KT
is also indispensable for revealing the basic properties of the estimators to be discussed in Section~\ref{sect:sectestimators}.
\begin{Thm} \label{thm:covrsrk}
	Let $\{(X_i ,Y_i)\}_{i=1}^n$, $S_1$ and $S_2$ be defined as in Lemma~\ref{lem:lemknown}. 
	Then the covariance between $\RS(X,Y)$ and $\RK(X,Y)$ is 
  \begin{align}
		  \hspace{-6pt}\cov(\RS,\RK)&{=}\frac{12}{n(n^2{-}1)}\Bigg[\frac{7n{-}5}{18}{+}(n{-}4)\frac{S_1^2}{\pi^2}{-}5(n{-}2)\frac{S_2^2}{\pi^2}\notag\\
	  	\label{eq:thmcovrsrk}
		  &\hspace{1cm}{-}6(n{-}2)^2\frac{S_1 S_2}{\pi^2}{+}(n{-}2)(n{-}3)\Omega_3(\rho)\Bigg]\\
	  	\label{eq:thmcovrsrkasymp}
		&{\simeq}\frac{12}{n}\left[\Omega_3(\rho){-}6\frac{S_1 S_2}{\pi^2}\right] \text{ (as } n \text{ large)} 
  \end{align}
where 
\begin{equation}
	\Omega_3(\rho)=\frac{1}{2}W_g+W_h.
	\label{eq:thmcovstatement}
\end{equation}
\end{Thm}
\begin{proof}
	See Appendix~\ref{app:appcovrsrk}.
\end{proof}
\begin{Rmk}
	The technique employed in  Appendix~\ref{app:appcovrsrk} can also provide an alternative proof of $\var(\RK)$ in (\ref{eq:lemvarrk}),
	by the relationship 
	\[
	\var({\RK})=\frac{1}{n^2(n-1)^2}\var(\mathcal{T})=\frac{1}{n^2(n-1)^2}\left[\mean(\mathcal{T}^2)-\mean^2(\mathcal{T})\right].
	\]
	The interested reader, after trying this, will find that the proof by this way is much simpler than the
	characteristic-function-based argument detailed in~\cite{kendall90}.
\end{Rmk}
\begin{Cly} \label{cor:covrsrk}
In Theorem~\ref{thm:covrsrk}, the covariance $\cov(\RS,\RK)$ can also be expressed as
 \begin{align}
	 \hspace{-6pt}\cov(\RS,\RK)&{=}\frac{12}{n(n^2{-}1)}\Bigg[\frac{(n{+}1)^2}{18}{+}(n{-}4)\frac{S_1^2}{\pi^2}{-}5(n{-}2)\frac{S_2^2}{\pi^2}\notag\\
	  	\label{eq:corcovrsrk}
		&\hspace{6.5mm}{-}6(n{-}2)^2\frac{S_1 S_2}{\pi^2}{+}\frac{2}{\pi^2}(n{-}2)(n{-}3)\Omega_4(\rho)\Bigg]\\
	  	\label{eq:corcovrsrkasymp}
		&{\simeq}\frac{12}{n}\left[\frac{1}{18}{+}2\frac{\Omega_4(\rho)}{\pi^2}{-}6\frac{S_1 S_2}{\pi^2}\right] \text{ (as } n \text{ large)} 
  \end{align}
  where
\begin{equation}
	\begin{split}
		\Omega_4(\rho)&=\int_0^\rho\left[\sin^{-1}\left(\frac{x}{3}\right)+2\sin^{-1}\left(\frac{x}{\sqrt{3}}\right)\right]\frac{dx}{\sqrt{1-x^2}}\\
		&{-}2\int_0^\rho\sin^{-1}\left(\frac{x}{2}\sqrt{\frac{1-x^2}{9-3x^2}}\right)\frac{dx}{\sqrt{4-x^2}}\\
		&{+}\int_0^\rho\sin^{-1}\left(\frac{x}{2}\frac{5-x^2}{3-x^2}\right)\frac{dx}{\sqrt{4-x^2}}\\
		&{-}2\int_0^\rho\sin^{-1}\left(x\sqrt{\frac{1-x^2}{12-6x^2}}\right)\frac{dx}{\sqrt{4-x^2}}\\
		&{+}2\int_0^\rho\sin^{-1}\left(x\sqrt{\frac{3-x^2}{4-2x^2}}\right)\frac{dx}{\sqrt{4-x^2}}.
	\end{split}
	\label{eq:corO4}
\end{equation}
\end{Cly}
\begin{proof}
Inverting (\ref{eq:p4}) yields
\begin{equation}
	W=16P_4^0-1-\frac{2}{\pi}\sum_{r=1}^{3}\sum_{s=r+1}^4\sin^{-1}\varrho_{rs}.
	\label{eq:WandP4}
\end{equation}
Combining (\ref{eq:thmcovstatement}) and (\ref{eq:WandP4}), $\Omega_3(\rho)$ can be rewritten in terms of $P_4^0$ and the correlation coefficients corresponding 
to $R_g$ and $R_h$ in Appendix 2 of~\cite{david1961}. This leads to 
\begin{equation}
	\Omega_3(\rho)=\frac{1}{18}+\frac{2}{\pi^2}\Omega_4(\rho).
	\label{eq:O4O3}
\end{equation}
The corollary thus follows directly by substituting (\ref{eq:O4O3}) to (\ref{eq:thmcovrsrk}) and (\ref{eq:thmcovrsrkasymp}), respectively. 
\end{proof}
\begin{Rmk}
	Both (\ref{eq:thmcovrsrk}) and (\ref{eq:corcovrsrk}) are \emph{exact} for any value of $n \ge 4$ and
	$|\rho|\le 1$. However, they are of different usefulness according to different numerical and analytical purposes.
	Formula (\ref{eq:thmcovrsrk}) is more convenient in the sence of
	controlling the precision of numerical integrations when programming; whereas (\ref{eq:corcovrsrk})
	is more convenient in the sence of evaluating any order ($\ge1$) of derivatives of $\cov(\RS,\RK)$ with respect to $\rho$.
	These higher order derivatives are mandatory when expanding $\cov(\RS,\RK)$ as a power series in $\rho$, 
	a conventional practice in the literature. For example, performing the Taylor expansion to 
	(\ref{eq:corcovrsrkasymp}) with the assistance of (\ref{eq:corO4}) gives 
	\begin{align}
		&\cov(\RS,\RK) \simeq \frac{2}{3n}\big(1-1.24858961\rho^2+0.06830496\rho^4\notag\\
		 \label{eq:covrsrkexp}
		 &+0.07280482\rho^6 + 0.04025528 \rho^8+0.02189277 \rho^{10}+\cdots\big)
	\end{align}
	which agrees with the formula (51) obtained in~\cite{dks1951}, except for the coefficients of the 
	last two terms, which we find to be $0.04025528$ and $0.02189277$, against their $0.04025526$
	and $0.01641362$, respectively. Since $\Omega_4(\rho)$ in (\ref{eq:corcovrsrk}) is exact , we believe that (\ref{eq:covrsrkexp}) 
	is more accurate than (51) in~\cite{dks1951}.  Unfortunately, even (\ref{eq:covrsrkexp}) is too coarse
	when $n$ is small and/or $|\rho|$ is large. To satisfy the requirments of the current study,
	we prefer to the $\Omega_3(\rho)$-based formula (\ref{eq:thmcovrsrk}), which can provide numerical results
	to any desired decimal place. For convenience of us as well as other researchers, a densely tabulated table, Table~\ref{tab:tabO3},  
	for $\Omega_3(\rho)$ with ten-place accuracy  is provided in Section~\ref{sect:sectnum}.  
\end{Rmk}
\begin{Rmk}
	Due to the complicated integrals involved in (\ref{eq:thmcovrsrk}) and (\ref{eq:corcovrsrk}), $\cov(\RS,\RK)$ cannot be expressed 
	in elementary functions. However, exact results are attainable for $\rho=0$ and $\rho=1$ ($-1$).
	It follows that (Appendix~\ref{app:appO1O2O3})
	\begin{align}
		\label{eq:O3rho0}
		\Omega_3(0)&=\frac{1}{18}\\
		\label{eq:O3rho1}
		\Omega_3(1)&=\frac{1}{2}.
	\end{align}
	Substituting (\ref{eq:O3rho0}) into (\ref{eq:thmcovrsrk}) yields
	\begin{equation}
		\cov(\RS,\RK)\big|_{\rho=0}=\frac{2}{3n}\frac{n+1}{n-1}
		\label{eq:covrsrkrho0}
	\end{equation}
	which is more readily to obtain on substitution of $\rho=0$ into (\ref{eq:corcovrsrk}).
	Regarding the case for $\rho=1$, it is rather difficult by means of substituting $\rho=1$ into (\ref{eq:corcovrsrk}) and 
	evaluating $\Omega_4(1)$ based on (\ref{eq:corO4}) thereafter. Fortunately, with the help of (\ref{eq:O3rho1}), it follows
	readily from (\ref{eq:thmcovrsrk}) and (\ref{eq:thmcovrsrkasymp}) that
	$	\cov(\RS,\RK)\big|_{\rho=1} = 0$
	which, again, is of no surprise but, to our knowledge, has never beed explictly proven in the literature.
	Due to symmetry, we also have   
	$	\cov(\RS,\RK)\big|_{\rho=-1} = 0.$
\end{Rmk}
\subsection{$\mean(\RS)$ and $\mean(\RK)$ in Contaminated Normal Model}
The PPMCC is notoriously sensitive to the non-Gaussianity caused by impulsive contamination in the data. 
Even a single outlier can distort severely the value of PPMCC and hence result in misleading inference 
in practice.  Assume that $(X,Y)$ obeys the following distribution~\cite{georgy2002}
\begin{equation}
	\bar{\epsilon}\mathcal{N}\left(\mu_X,\mu_Y,\sigma_X^2,\sigma_Y^2,\rho\right)+\epsilon\mathcal{N}\left(\mu_X,\mu_Y,\lambda_X^2\sigma_X^2,\lambda_Y^2\sigma_Y^2,\rho'\right)	
	\label{eq:gmm}
\end{equation}
where $\bar{\epsilon}\define 1-\epsilon$, $0\le\epsilon\le 1$, $\lambda_X\gg1$, and $\lambda_Y\gg1$. 
Under this  Gaussian contamination model that frequently used in the literature of
robustness analysis~\cite{stein1995,wang2000it,reznic2002}, it has been shown that,
no matter how small $\epsilon$ is, the expectation of the PPMCC $\mean(\RP)\to\rho'$ as 
$\lambda_X\to\infty$ and $\lambda_Y\to\infty$~\cite{georgy2002}. On the other hand, 
as shown in the theorem below, SR and KT are more robust than PPMCC under the model (\ref{eq:gmm}).

\begin{Thm} \label{thm:thmrobust}
	Let $\{(X_i,Y_i)\}_{i=1}^n$ be i.i.d. samples generated from the model (\ref{eq:gmm}).  
	Let $\RS$ and $\RK$ be the SR and KT defined in (\ref{eq:rsdef}) and (\ref{eq:rkdef}), respectively.
Then
\begin{align}
	\label{eq:thmgmmmeanrk}
	\lim_{\substack{\epsilon\to 0\\\lambda_X\to\infty\\\lambda_Y\to\infty}}\mean(\RK)&=\frac{2}{\pi}\left[(1-2\epsilon)\sin^{-1}\rho+2\epsilon\sin^{-1}\rho'\right]\\
	\label{eq:thmgmmmeanrs}
	\lim_{\substack{\epsilon\to 0\\ n\to\infty\\\lambda_X\to\infty\\\lambda_Y\to\infty}}\mean(\RS)&=\frac{6}{\pi}\left[(1-3\epsilon)\sin^{-1}\frac{\rho}{2}+\epsilon\sin^{-1}\rho'\right].
\end{align}
\end{Thm}
\begin{proof}
	See Appendix~\ref{app:appgmm}.
\end{proof}

\begin{Rmk}
	It was stated without substantial argument in~\cite{georgy2002} that, under the model (\ref{eq:gmm}), $\mean(\RS)$ is of the following form
	\begin{equation}
		\mean(\RS)=\frac{6}{\pi}\left[(1-\epsilon)\sin^{-1}\frac{\rho}{2}+\epsilon\sin^{-1}\frac{\rho'}{2}\right] \tag{$\ast$}
	\end{equation} 
as $\epsilon\to0$, $\lambda_X\to\infty$ and $\lambda_Y\to\infty$.
This is quite inconsistent with our result (\ref{eq:thmgmmmeanrs}) in Theorem~\ref{thm:thmrobust}. 
We will resolve the controversy between  (\ref{eq:thmgmmmeanrs}) and ($\ast$) by Monte Carlo simulations in Section~\ref{sect:sectnum}. 
\end{Rmk}

\section{Estimators of the Population Correlation} \label{sect:sectestimators}
In this section, we investigate the performance of the estimators of $\rho$
based on SR and KT in terms of bias, MSE and ARE to PPMCC. 
To gain further insight into their relationship,
the correlation between the two estimators $\hat\rho_S$ and $\hat\rho_K$ (defined below) is also derived. 

\subsection{Asymptotic Unbiased Estimators} 
Inverting (\ref{eq:lemmeanrp}), (\ref{eq:lemmeanrsasymp}) and (\ref{eq:lemmeanrk}), we have the following estimators of $\rho$ 
\begin{align}
	\label{eq:estp}
	\hat\rho_P & \define \RP \\
	\label{eq:ests}
	\hat\rho_S &\define 2 \sin\left(\frac{\pi}{6}\RS\right)\\
	\label{eq:estk}
	\hat\rho_K &\define  \sin\left(\frac{\pi}{2}\RK\right).
\end{align}
Moreover, another estimator based on a mixture of $\RS$ and $\RP$ can be constructed as~\cite{kendall90}
\begin{equation}
	\hat\rho_M \define 2\sin\left(\frac{\pi}{6}\RS-\frac{\pi}{2}\frac{\RK-\RS}{n-2}\right)
	\label{eq:estm}
\end{equation}
based on the following relationship
\begin{equation}
	\mean(\RS)=\frac{6}{\pi}\left(S_2+\frac{S_1-3S_2}{n+1}\right).
	\label{eq:eqrelation}
\end{equation}

In the sequel we will focus on the properties of the estimators defined in (\ref{eq:estp})--(\ref{eq:estm}).
Here the estimator $\hat\rho_P$ in (\ref{eq:estp}) is employed as a benchmark due to its optimality for normal samples, 
in the sense of approaching the CRLB~\cite{kendall91} when the sample size is sufficiently large.
\subsection{Bias Effect for Small Samples}
It is noteworthy that the four estimators in (\ref{eq:estp})--(\ref{eq:estm}) are unbiased only for large samples.
When the sample size is small, the bias effects, as shown in the following theorem, are not ignorable any more.
\begin{Thm} \label{thm:bias}
	Let $\hat\rho_\zeta$, $\zeta\in\{P,S,K,M\}$ be defined as in (\ref{eq:estp})--(\ref{eq:estm}), respectively.
	Define $\bias_{\zeta}\define\mean(\hat\rho_{\zeta}-\rho)$. Let $S_1$ and $S_2$ bear the same meanings as in
	Lemma~\ref{lem:lemknown}.
	Write $\sigma_S^2\define \var(\RS)$, $\sigma_K^2\define\var(\RK)$ and $\sigma_{S,K}\define\cov(\RS,\RK)$. 
	Then, under the same assumptions made as in Theorem~\ref{thm:varrs},  
	\begin{align}
		\label{eq:thmbiasp}
		\bias_P& \simeq -\frac{1}{2n}\rho(1-\rho^2)\\
		\label{eq:thmbiass}
		\bias_S& \simeq \frac{\sqrt{4-\rho^2}}{n+1}\left(S_1-3S_2\right)-\frac{\pi^2\rho}{72}\sigma_S^2\\
		\label{eq:thmbiask}
		\bias_K& \simeq -\frac{\pi^2\rho}{8}\sigma_K^2\\
		\label{eq:thmbiasm}
		\bias_M& \simeq -\frac{1}{72}\frac{\pi^2\rho}{(n{-}2)^2}\left[(n{+}1)^2\sigma_S^2{-}6(n{+}1)\sigma_{S,K}{+}9\sigma_K^2\right]. 
	\end{align}
\end{Thm}
\begin{proof}
	The first statement (\ref{eq:thmbiasp}) follows directly from (\ref{eq:lemmeanrp}) in Lemma~\ref{lem:lemknown}.
	Now we proceed to evaluate $\bias_S$, $\bias_K$ and $\bias_M$. For convenience, 
	write $\overline{\RS} {\define}\mean(\RS)$, $\overline{\RK} {\define}\mean(\RK)$, 
	$\delta_S{\define}\RS{-}\overline{\RS}$, and $\delta_K{\define}\RK{-}\overline{\RK}$.
	Expanding (\ref{eq:ests}) around $\overline{\RS}$ yields 
	\begin{equation}
		\hat\rho_S{=}2 \sin\left(\frac{\pi}{6}\overline{\RS}\right){+}\frac{\pi}{3}\cos\left(\frac{\pi}{6}\overline{\RS}\right)\delta_S
		{-} \frac{\pi^2}{36}\sin\left(\frac{\pi}{6}\overline{\RS}\right)\delta_S^2{+}\cdots.
		\label{eq:estsexp}
	\end{equation}
	Taking expectation of both sides in (\ref{eq:estsexp}), applying $\mean(\delta_S)=0$, $\mean(\delta_S^2)=\sigma_S^2$
	and ignoring the high order infinitesimals, we have
	\begin{equation}
		\mean(\hat\rho_S)\simeq 2 \sin\left(\frac{\pi}{6}\overline{\RS}\right)- \frac{\pi^2}{36}\sin\left(\frac{\pi}{6}\overline{\RS}\right)\sigma_S^2.
		\label{eq:estsmean}
	\end{equation}
	Substituting (\ref{eq:eqrelation}) into (\ref{eq:estsmean}),
	expanding to the order of $(n+1)^{-1}$, and subtracting $\rho$ thereafter, 
	we obtain the result  (\ref{eq:thmbiass}). In a similar way we have
	\begin{equation*}
		\mean(\hat\rho_K)\simeq\rho-\frac{\pi^2\rho}{8}\sigma_K^2
	\end{equation*}
	which leads directly to (\ref{eq:thmbiask}). Performing Taylor expansion of $\hat\rho_M(\RS,\RK)$
	around $(\overline\RS,\overline\RK)$ till the second order, we have
	\begin{equation}
	\begin{split}
		\hat\rho_M&=\hat\rho_M(\overline\RS,\overline\RK)+
		\frac{\partial(\hat\rho_M)}{\partial(\RS)}\delta_S+
		\frac{\partial(\hat\rho_M)}{\partial(\RK)}\delta_K\\
		&+\frac{1}{2}\left[ \frac{\partial^2(\hat\rho_M)}{\partial(\RS)^2}\delta_S^2{+}\frac{\partial^2(\hat\rho_M)}{\partial(\RK)^2}\delta_K^2
		{+}\frac{2\partial^2(\hat\rho_M)\delta_S\delta_K}{\partial(\RS)\partial(\RK)}\right]+\cdots.
	\end{split}
		\label{eq:eqtaylorrhom}
	\end{equation}
	Taking expectation of both sides in (\ref{eq:eqtaylorrhom}), ignoring high order infinitesimals, applying the results 
	$\hat\rho_M(\overline\RS,\overline\RK)=\rho$,  $\mean(\delta_S)=0$,  $\mean(\delta_K)=0$, $\mean(\delta_S^2)=\sigma_S^2$, $\mean(\delta_K^2)=\sigma_K^2$,
	$\mean(\delta_S,\delta_K)=\sigma_{S,K}$ along with the second order partial derivatives
	\begin{align*}
	\frac{\partial^2\hat\rho_M(\overline\RS,\overline\RK)}{\partial(\RS)^2}&=-\frac{\pi^2\rho}{36}\frac{(n+1)^2}{(n-2)^2}\\
	\frac{\partial^2\hat\rho_M(\overline\RS,\overline\RK)}{\partial(\RK)^2}&=-\frac{\pi^2\rho}{4}\frac{1}{(n-2)^2}\\
	\frac{\partial^2\hat\rho_M(\overline\RS,\overline\RK)}{\partial(\RS)\partial(\RK)}&=\frac{\pi^2\rho}{12}\frac{n+1}{(n-2)^2}
\end{align*}
and subtracting $\rho$ thereafter, we arrive at the forth theorem statement (\ref{eq:thmbiasm}), thus completing the proof. 
\end{proof}
\begin{Rmk}
From (\ref{eq:thmbiasp})--(\ref{eq:thmbiasm}), it follows that, for all the four estimators,
\begin{itemize}
	\item  $\bias_\zeta(\rho)=\bias_\zeta(-\rho)$ (odd symmetry);
	\item  $\rho\bias_\zeta(\rho)\le 0$ (negative bias); 
	\item  $\bias_\zeta=0$ for $\rho\in\{-1,0,1\}$;
	\item  $\bias_\zeta\sim O(n^{-1})$ as $n\to\infty$.
\end{itemize}
Moreover, contrary to $\bias_P$ and $\bias_K$, $\bias_S$ and $\bias_M$ cannot be expressed into  
elementary functions due to the intractability involved in (\ref{eq:thmvarrs}) and (\ref{eq:thmcovrsrk}), the expressions of $\var(\RS)$ and $\cov(\RS,\RK)$, respectively.
\end{Rmk}

\subsection{Approximation of Variances} 
Besides the bias effect just discussed, the variance is another important \emph{figure of merit} when comparing the performance of the estimators $\hat\rho_\zeta$, $\zeta\in\{P,S,K,M\}$.
From (\ref{eq:lemmeanrp}), it follows that 
\begin{equation}
	\var(\hat\rho_P)\simeq\frac{(1-\rho^2)^2}{n-1}.
	\label{eq:varrhop}
\end{equation}
By the {delta method}, it follows that~\cite{kendall90} 
	\begin{align}
		\label{eq:varrhos}
		\var(\hat\rho_S)&\simeq\frac{\pi^2(4-\rho^2)}{36}\var(\RS)\\
		\label{eq:varrhok}
		\var(\hat\rho_K)&\simeq\frac{\pi^2(1-\rho^2)}{4}\var(\RK).
	\end{align} 
Now we only need to deal with $\var(\hat\rho_M)$, which is stated below.
\begin{Thm} \label{thm:mse}
	Let $\hat\rho_M$ be defined as in (\ref{eq:estm}).
	Then, under the same assumptions made as in Theorem~\ref{thm:varrs},  
	\begin{equation}
		\var(\hat\rho_M)\simeq\frac{\pi^2(4-\rho^2)}{36(n-2)^2}\left[(n+1)^2\sigma_S^2-6(n+1)\sigma_{S,K}+9\sigma_K^2\right]. 
		\label{eq:varrhom}
	\end{equation}
\end{Thm}
\begin{proof}
	Using the delta method~\cite{kendall91}, it follows that
\begin{equation}
	\var(\hat\rho_M){\simeq} \left[\frac{\partial(\hat\rho_M)}{\partial(\RS)}\right]^2\hspace{-2mm}\sigma_S^2{+} \left[\frac{\partial(\hat\rho_M)}{\partial(\RK)}\right]^2\hspace{-2mm}\sigma_K^2
	{+}2\frac{\partial(\hat\rho_M)}{\partial(\RS)}\frac{\partial(\hat\rho_M)}{\partial(\RK)}\sigma_{S,K}.
	\label{eq:thmvarrhom}
\end{equation}
	The theorem thus follows with substitutions of the partial derivatives 
	\begin{align*}
		\frac{\partial\hat\rho_M(\overline\RS,\overline\RK)}{\partial(\RS)}&=\frac{\pi}{6}\frac{n+1}{n-2}\sqrt{4-\rho^2}\\
		\frac{\partial\hat\rho_M(\overline\RS,\overline\RK)}{\partial(\RK)}&=\frac{\pi}{2}\frac{-1}{n-2}\sqrt{4-\rho^2}
	\end{align*}
	into (\ref{eq:thmvarrhom}), respectively.
\end{proof}
\subsection{Asymptotic Relative Efficiency} 
Thus far in this section we have established the analytical results  with an emphasis 
on limited-sized bivariate normal samples. For a better understanding of the fourt estimators, we will shift our attention
to the asymptotic properties of $\hat\rho_\zeta$ in the sequel. Since $\lim_{n\to\infty} \mean(\hat\rho_\zeta) =\rho$, 
we can compare their performances by means of the \emph{asymptotic relative efficiency}, which is defined as~\cite{kendall91}
\begin{equation} \label{eq:aredef}
	\ARE_\zeta\define \lim_{n\rightarrow\infty}\frac{\var(\hat \rho_P)}{\var(\hat \rho_\zeta)},\quad\zeta\in\{P,S,K,M\}. 
\end{equation}
As remarked before, we employ $\hat\rho_P$ as a benchmark, since, for large-sized bivariate normal samples,  $\hat\rho_P$  approaches 
the Cramer-Rao lower bound (CRLB)~\cite{kendall91}
\begin{equation} \label{eq:crlb}
	\mathrm{CRLB}=\frac{(1-\rho^2)^2}{n}.
\end{equation}
From (\ref{eq:aredef}) it is obvious that $\ARE_P=1$. 
Moreover, comparing  (\ref{eq:varrhos}) and (\ref{eq:varrhom}), it is easily seen that 
$
\lim_{n\to\infty} \var(\hat\rho_S)/\var(\hat\rho_M)=1,
$
which leads readily to $\ARE_S=\ARE_M$ by referring to (\ref{eq:aredef}). 
 Then we only need to focus on $\ARE_S$ and $\ARE_K$ in the following discussion.
\begin{Thm} \label{thm:are}
	Let $\ARE_S$ and $\ARE_K$ be defined as in (\ref{eq:aredef}). Then
	\begin{align}
		\label{eq:thmares}
\ARE_S&= \frac{36(1-\rho^2)^2}{(4-\rho^2)\left[9\pi^2\Omega_1(\rho)-324\left(\sin^{-1}\frac{1}{2}\rho\right)^2\right]}\\
		\label{eq:thmarek}
		\ARE_K &= \frac{9(1-\rho^2)}{\pi^2-36\left(\sin^{-1}\frac{1}{2}\rho\right)^2}.
	\end{align}
\end{Thm}
\begin{proof}
Substituting  (\ref{eq:varrhos}) and (\ref{eq:varrhok}) into (\ref{eq:aredef}) 
yields (\ref{eq:thmares}) and (\ref{eq:thmarek}), respectively,  and the proof completes.
\end{proof}
\begin{Rmk} Due to the intractability of $\Omega_1(\rho)$ in (\ref{eq:thmares}), $\ARE_S$ cannot
	be expressed into elementary functions in general. However, exact results are obtainable for 
	$\rho=0,\pm1$. Substituting $\rho=0$ and $\Omega_1(0)=1/9$ into (\ref{eq:thmares}), 
	it is easy to verify that
	\[
	\ARE_S(0)=\frac{9}{\pi^2}\simeq 0.9119  
	\]
	which is a well known result~\cite{kendall90}. In our previous work \cite{xu2010} we also obtained that
	\begin{equation}
		\label{eq:arespm1}
	\ARE_S(\pm1)=\frac{15+11\sqrt{5}}{57}\simeq 0.6947.
	\end{equation}
	Now let us investigate $\ARE_K$. It follows from (\ref{eq:thmarek})
	that, $\ARE_K$ is expressible as elementary functions of $\rho$, and is therefore more tractable than $\ARE_S$.
	In other words, we can evaluate easily any value of $\ARE_K$ with respect to any value of $\rho \ne \pm 1$.
	For example, substituting $\rho=0$ into (\ref{eq:thmarek}) yields
	\[
	\ARE_K(0)=\frac{9}{\pi^2}
	\]
	which is identical to $\ARE_S(0)$ and also well known~\cite{kendall90}. However, when $\rho\to\pm1$, an extra effort is necessary,
	since both the numerator and denominator of (\ref{eq:thmarek}) vanish in this case.
	Apply the L'Hopital's rule, we find the following result
	\begin{equation}
		\ARE_K\big|_{\rho\to\pm1}=\frac{1}{4}\frac{\rho\sqrt{4-\rho^2}}{\sin^{-1}\frac{1}{2}\rho}\Bigg|_{\rho=\pm1}
		=\frac{3\sqrt{3}}{2\pi}\simeq 0.8270
		\label{eq:arekpm1}
	\end{equation}
	which is greater than $\ARE_S(\pm1)$. In fact, a comparison of $\ARE_S$ and $\ARE_K$ in Section~\ref{sect:sectnum}  
	suggest that $\ARE_S\le \ARE_K$ for all $\rho\in[-1,1]$.  	
\end{Rmk}
\begin{table*}
	\centering
	\caption{Values of $\Omega_1(\rho)$ and $\Omega_2(\rho)$ in Theorem~\ref{thm:varrs}}.
	\includegraphics{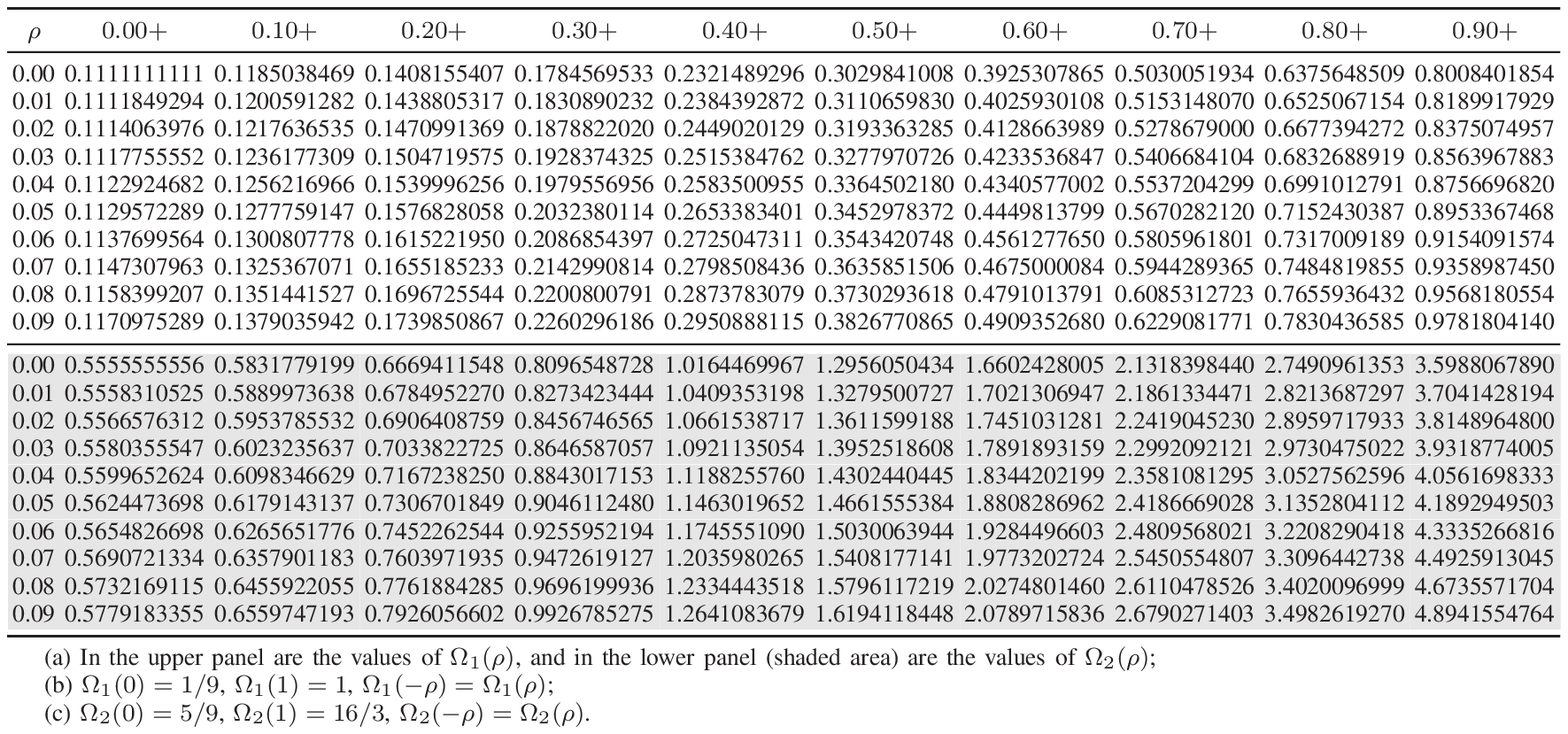}
\label{tab:tabO1O2}
\end{table*}
\begin{table*}
	\centering
	\caption{Values of $\Omega_3(\rho)$ in Theorem~\ref{thm:covrsrk}}
	\includegraphics{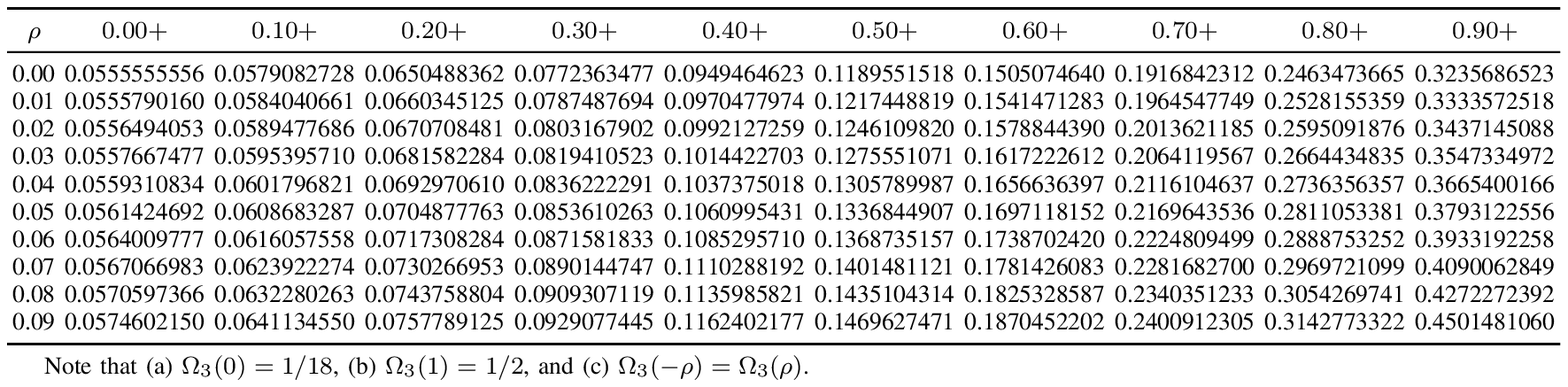}
\label{tab:tabO3}
\end{table*}
\section{Numerical Results} \label{sect:sectnum}
In this section we aim at
1) tabulating the values of $\Omega_1(\rho)$, $\Omega_2(\rho)$ (in Theorem~\ref{thm:varrs}) and $\Omega_3(\rho)$ (in Theorem~\ref{thm:covrsrk})
that are not expressible as elementary functions,
2) verifying the theoretical results Theorems~\ref{thm:varrs} to \ref{thm:are} established in previous sections, and
3) comparing the performances of the four estimators defined in (\ref{eq:estp})--(\ref{eq:estm})
by means of \emph{bias effect}, \emph{mean square error} (MSE) and  ARE  under both the normal and contaminated normal models.
Throughout this section, Monte Carlo experiments are undertaken for $10 \le n\le 100$.
A sample size of $n=1000$ is considered large enough when we investigate the asymptotic behaviors.
The number of trials is set to $5\times10^5$ for reason of accuracy.  
All samples are  generated by functions in the Matlab $\mathtt{Statistics\; Toolbox}^{\mathtt{TM}}$.
Specifically, the normal samples are generated by  \texttt{mvnrnd}, whereas the contaminated normal samples
are generated by \texttt{gmdistribution} and \texttt{random}.  
The notation $\rho=\rho_1(\Delta\rho)\rho_2$ represents a list of $\rho$
starting from $\rho_1$ to $\rho_2$ with increment $\Delta\rho$.
\subsection{Tables of $\Omega_1(\rho)$, $\Omega_2(\rho)$ and $\Omega_3(\rho)$}
Table~\ref{tab:tabO1O2} contains the values of $\Omega_1(\rho)$ and $\Omega_2(\rho)$  in 
(\ref{eq:thmvarrs}), the first statement of Theorem~\ref{thm:varrs} for $\rho=0(0.01)1$.
In the upper panel are the values of $\Omega_1(\rho)$; 
whereas in the lower panel are the values of $\Omega_2(\rho)$. 
Due to the importance of $\var(\RS)$ both in theory and in practice,
the table is made as intensive and accurate as possible, with the increment  $\Delta\rho$ being $0.01$,
and the precision being up to ten decimal places.  
In Table~\ref{tab:tabO3} are tabulated the values of $\Omega_3(\rho)$ in (\ref{eq:thmcovrsrk}) 
of Theorem~\ref{thm:covrsrk} for $\rho=0(0.01)1$. Because of the similar reasons,
the increment $\Delta\rho$ and precision are the same as those in Table~\ref{tab:tabO1O2}.
The values of $\Omega_1(\rho)$, $\Omega_2(\rho)$ and $\Omega_3(\rho)$ with repect to $\rho$ not included in
Tables~\ref{tab:tabO1O2} and \ref{tab:tabO3} can be easily obtained by interpolation. 
Given these tables, we can easily calculate the quantities that depend on $\Omega_1(\rho)$, $\Omega_2(\rho)$ and $\Omega_3(\rho)$,
including $\var(\RS)$, $\var(\hat\rho_S)$, $\var(\hat\rho_M)$, $\bias(\hat\rho_S)$, $\bias(\hat\rho_M)$, 
$\ARE_S$, $\ARE_M$, and so forth. 

\begin{figure}
	\centering
	\includegraphics[width=\columnwidth]{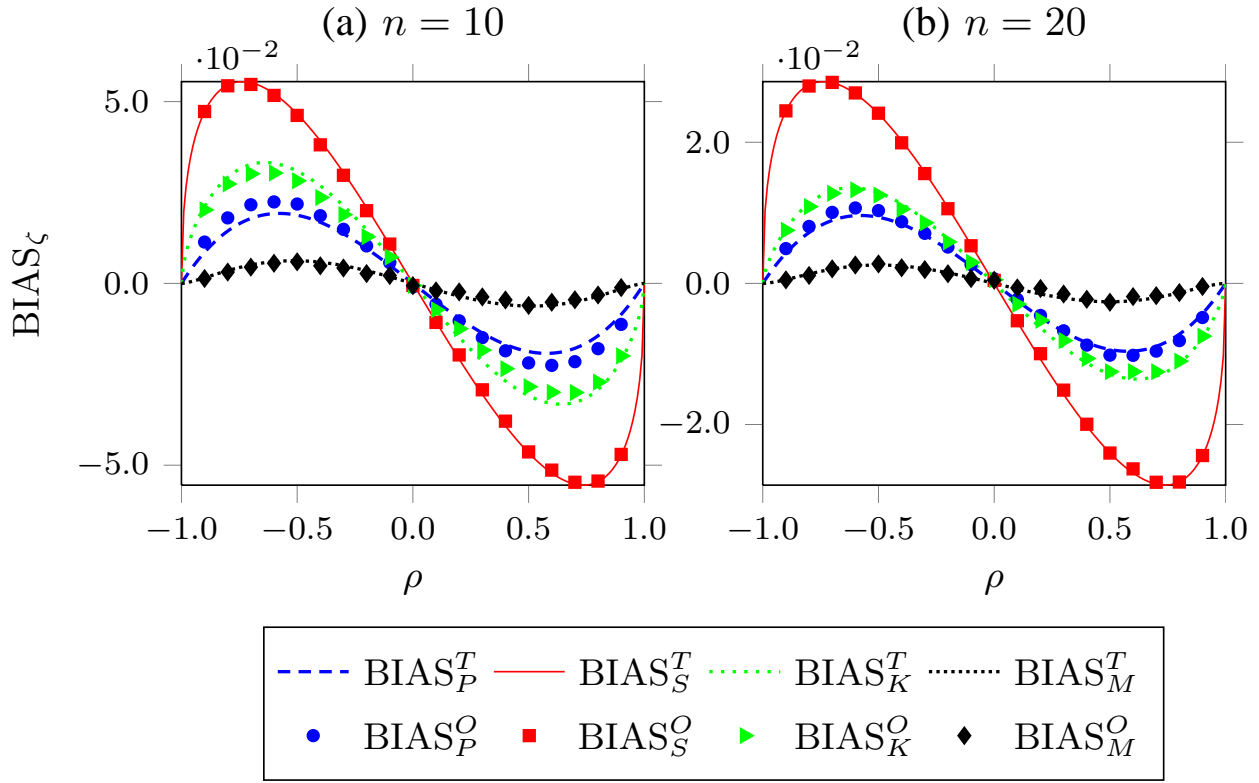}
	\caption{Comparison of $\bias_\zeta$, $\zeta\in\{P,S,K,M\}$ for (a) $n=10$ and (b) $n=20$. Theoretical curves,
	denoted by $\bias_\zeta^T$ in the legend, are plotted  over $\rho=-1(0.01)1$ based on (\ref{eq:thmbiasp})--(\ref{eq:thmbiasm}), respectively;
	whereas the simulation results, denoted by $\bias_\zeta^O$ in the legend, are plotted over $\rho=-0.9(0.1)0.9$.}
\label{fig:figbias}
\end{figure}
\begin{figure}
	\includegraphics[width=\columnwidth]{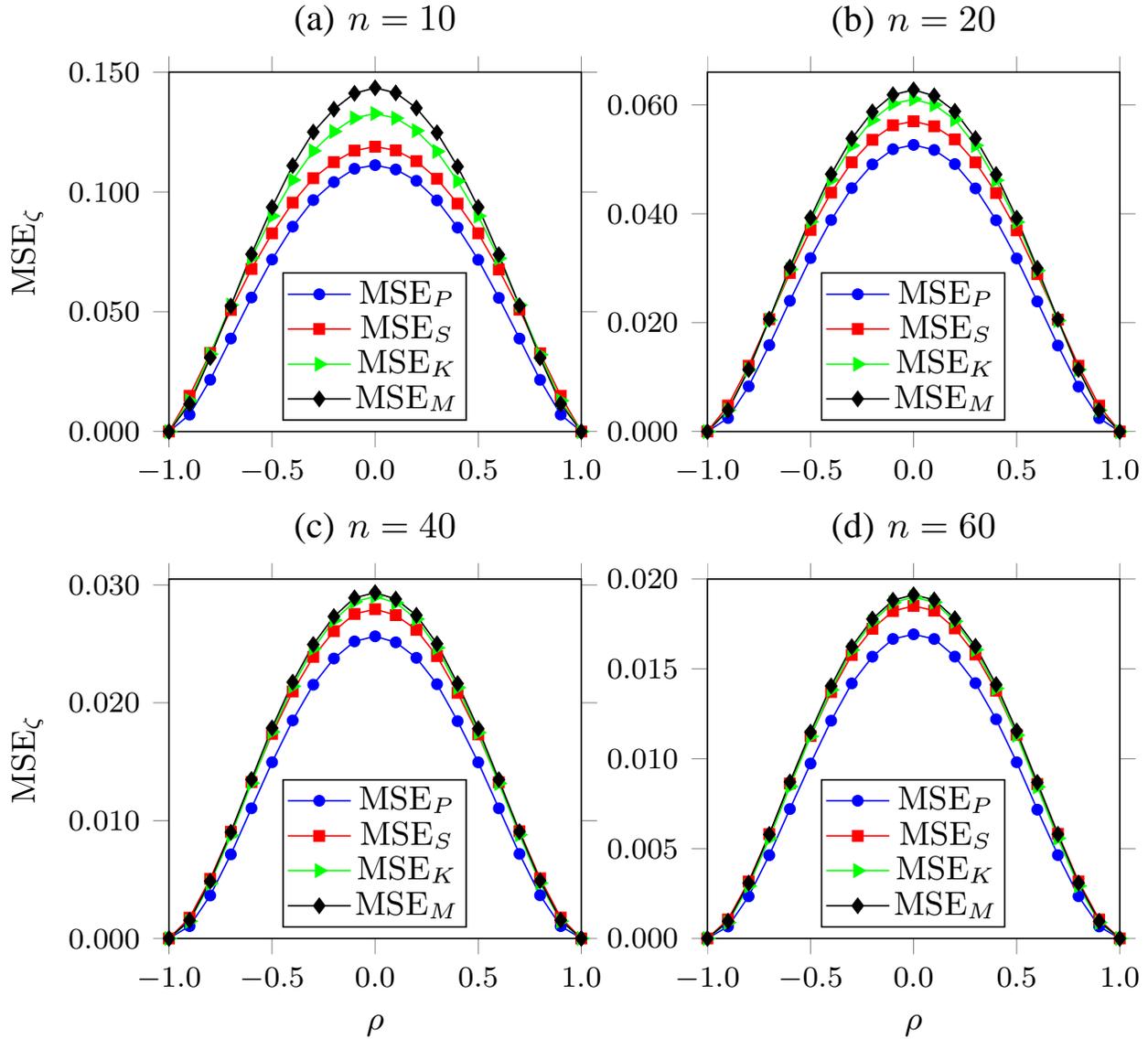}
	\caption{Comparison of observed $\mse_\zeta$, $\zeta\in\{P,S,K,M\}$ for (a) $n=10$, (b) $n=20$,
	(c) $n=40$, and (d) $n=60$ over $\rho=-1(0.1)1$, respectively.} 
	\label{fig:figmse}
\end{figure}
\begin{figure}
	\centering
	\includegraphics[width=\columnwidth]{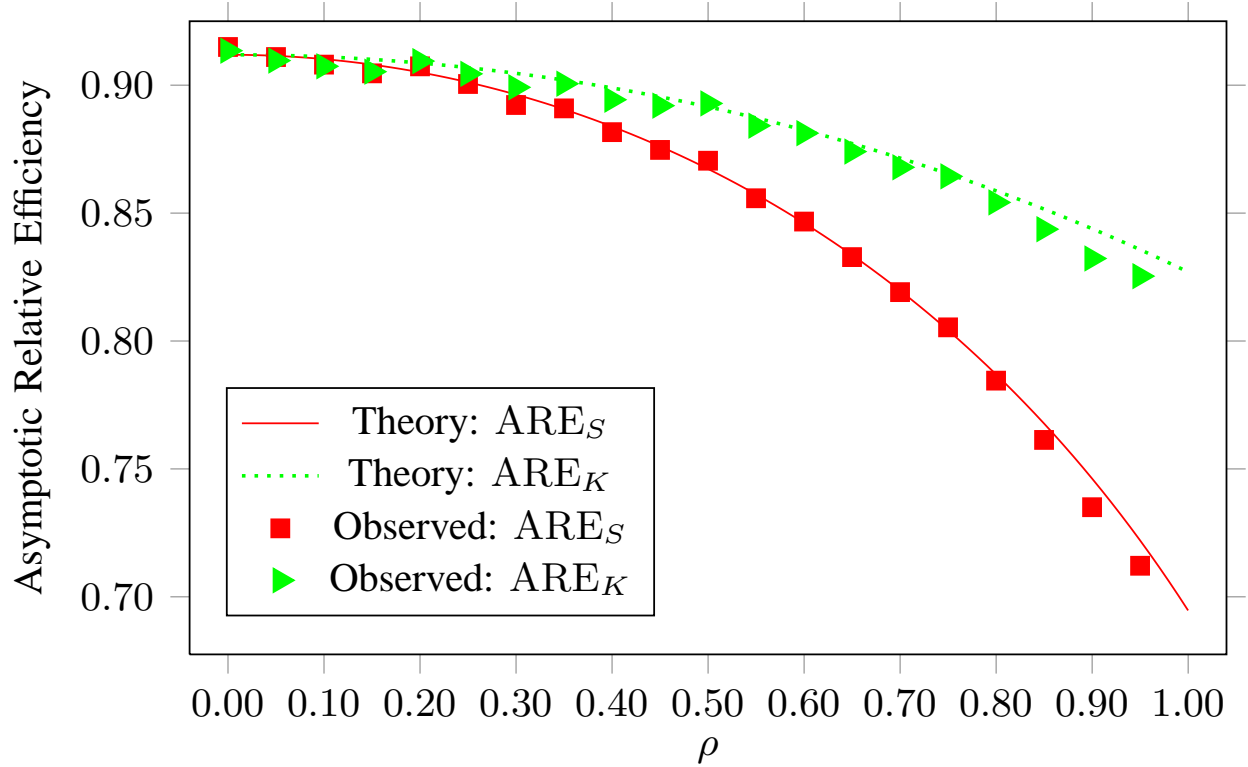}
	\caption{Verification and Comparison of $\ARE_S$ and $\ARE_K$ for $n=1000$  over $\rho=0(0.01)1$, for theoretical curves, 
	and $\rho=0(0.05)0.95$, for simulation results. 
	The results (\ref{eq:arespm1}) and (\ref{eq:arekpm1})
	are used to plot the two theoretical curves for $\rho=1$, respectively.
	}
\label{fig:figare}
\end{figure}
\begin{table*}
	\centering
	\caption{Variances of $\var(\hat\rho_\zeta)$, $\zeta\in\{P,S,K,M\}$ for $n=10,20,30$}
	\includegraphics[width=\textwidth]{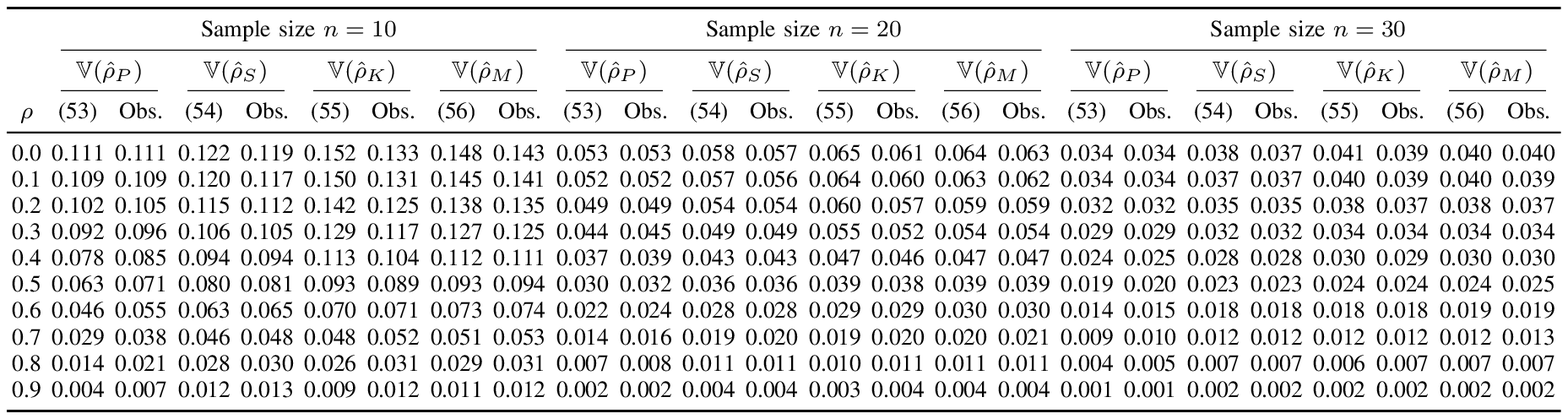}
	\label{tab:tabvar}
\end{table*}
\subsection{Verification of $\bias_\zeta$ and $\var(\hat\rho_\zeta)$ in Small Samples}
Fig.~\ref{fig:figbias} shows the bias effects $\bias_\zeta$ corresponding to the four estimators $\rho_\zeta$, $\zeta\in\{P,S,K,M\}$
for $n=10$ and $n=20$, respectively. It is clearly observed that the magnitudes of $\bias_\zeta$ can be ordered as $\bias_M<\bias_P<\bias_K<\bias_S$.
That is, the performance of $r_S$ is much worse than those of the other three in terms of bias effect in small samples. 
Moreover, it is also observed that  (\ref{eq:thmbiass}) and (\ref{eq:thmbiasm}) with respect to $\bias_S$ and $\bias_M$ 
are more accurate than (\ref{eq:thmbiasp}) and (\ref{eq:thmbiask}) with respect to $\bias_P$ and $\bias_K$.
In other words, the former two formulae agree better than do the latter two  formulae with the corresponding 
simulation results for a sample size $n$ as small as $10$. 
Nevertheless,  the deviations from (\ref{eq:thmbiasp}) and (\ref{eq:thmbiask}) to the corresponding simulation results
are less noticeable when the sample size $n$ is increased up to $20$. 

Table~\ref{tab:tabvar} lists, for each of the three samples sizes, $10$, $20$ and $30$, 
1) the theoretical results (\ref{eq:varrhop})--(\ref{eq:varrhom}) with respect to $\var(\hat\rho_\zeta)$ and
2) the corresponding observed variances from the Monte Carlo simulations. 
It can be seen that (\ref{eq:varrhos}) and (\ref{eq:varrhom}), with respect to $\var(\hat\rho_S)$ and $\var(\hat\rho_M)$,
are accurate enough even though the sample size is as small as $n=10$. 
On the other hand, unfortunately, the theoretical formula (\ref{eq:varrhop}) for $\var(\hat\rho_P)$ 
and (\ref{eq:varrhok}) for $\var(\hat\rho_K)$ deviate substantially from the corresponding 
observed simulation results for the same sample size $n=10$.
However, it appears that these deviations become less noticeable for $n=20$ and negligible for $n=30$.
Therefore, it would be save to use (\ref{eq:varrhop})--(\ref{eq:varrhom}) 
when approximating the variances of $\hat\rho_\zeta$  for $n\ge 20$  in practice.

\subsection{Comparison of MSE in Small Samples}
Contrary to $\bias_\zeta$ illustrated in Fig.~\ref{fig:figbias}, 
the magnitudes of the mean square errors 
\[
\mse_\zeta\define\mean\left[(\hat\rho_\zeta-\rho)^2\right],\, \zeta\in\{P,S,K,M\}
\] 
cannot be ordered in a consistent manner. It appears in Fig.~\ref{fig:figmse} that 
1) $\mse_M>\mse_K>\mse_S>\mse_P$ when $|\rho|$ is around $0$,
2) $\mse_S>\mse_K>\mse_P$ when $|\rho|$ exceeds some threshold, which moves towards $0$ with increase of $n$, 
and 3) the difference between $\mse_K$ and $\mse_S$ around $\rho=0$ decreases steadily with increase of $n$.
Furthermore, due to the asymptotic equivalence between $\hat\rho_S$ and $\hat\rho_M$, 
$\mse_S$ and $\mse_M$ becomes closer and closer as  $n$ increases.

\subsection{Verification and Comparison of $\ARE_S$ and $\ARE_K$}
Fig.~\ref{fig:figare} verifies and compares the performance of $\hat\rho_S$ and $\hat\rho_K$ in terms of ARE. 
For purpose of numerical verification, simulation results for $n=1000$
are superimposed upon the corresponding theoretical curves. Due to the asymptotic equivalence between $\hat\rho_S$ and $\hat\rho_M$,
the results with respect to $\ARE_M$ are not included in Fig.~\ref{fig:figare}.
It can be observed that 
1) the simulations agree well with our theoretical findings in (\ref{eq:thmares}) and (\ref{eq:thmarek}), respectively,
2) $\ARE_K$ lies consistently above $\ARE_S$, indicating the superiority of $\hat\rho_K$ over $\hat\rho_S$ for large samples, and  
3) the performance of $\hat\rho_S$ deteriorates severely as $\rho$ approaching unity, although it performs similarly as $\hat\rho_K$ when $\rho$ is small.     
Note that the remarks on $\ARE_S$ also apply to $\ARE_M$ due to the equivalence between  $\hat\rho_S$ and $\hat\rho_M$ when the sample size $n$ is large.

\begin{figure}
\includegraphics[width=\columnwidth]{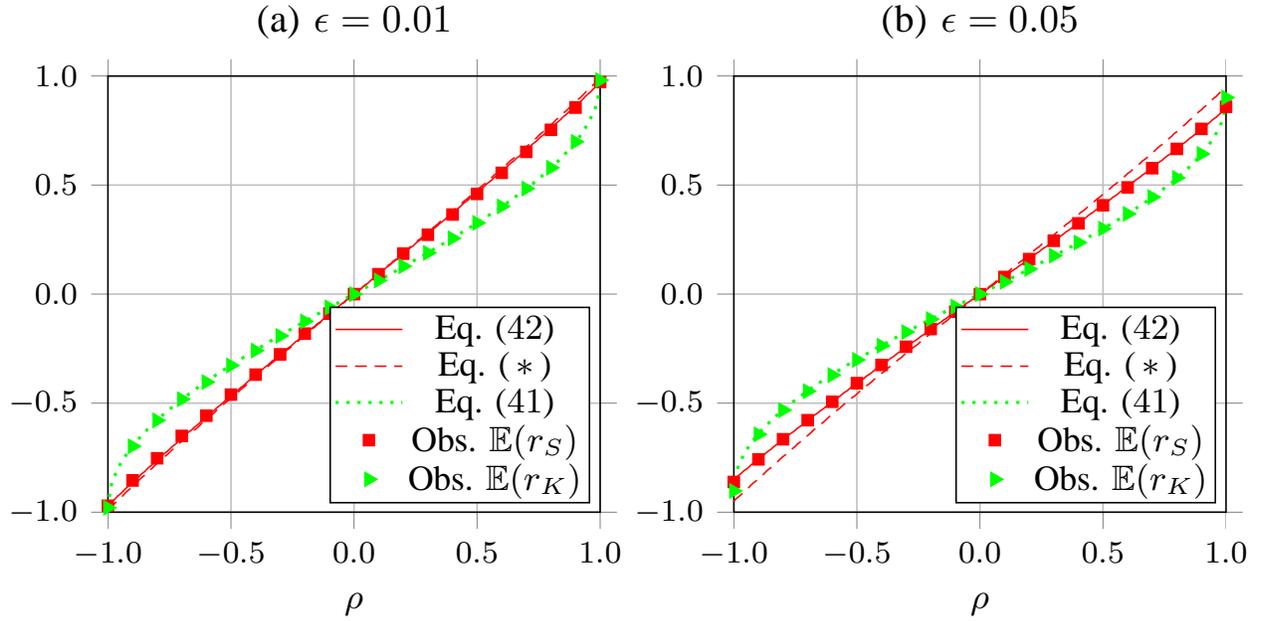}
	\caption{Verification of Theorem~\ref{thm:thmrobust} for (a) $\epsilon=0.01$ and (b) $\epsilon=0.05$
	under the sample size $n=50$ over $\rho=(-1)0.1(1)$, for simulations, and $\rho=(-1)0.01(1)$, for theoretical formulae (\ref{eq:thmgmmmeanrk})
	and (\ref{eq:thmgmmmeanrs}).
	The rest parameters of the model (\ref{eq:gmm}) are set to be $\sigma_X=\sigma_Y=1$, $\lambda_X=\lambda_Y=100$ and $\rho'=0$, respectively.
	The formula  ($*$) concerning $\mean(\RS)$ developed elsewhere~\cite{georgy2002} is also included for comparison.}
\label{fig:figgmmrsrk}
\end{figure}
\begin{figure}
	\includegraphics[width=\columnwidth]{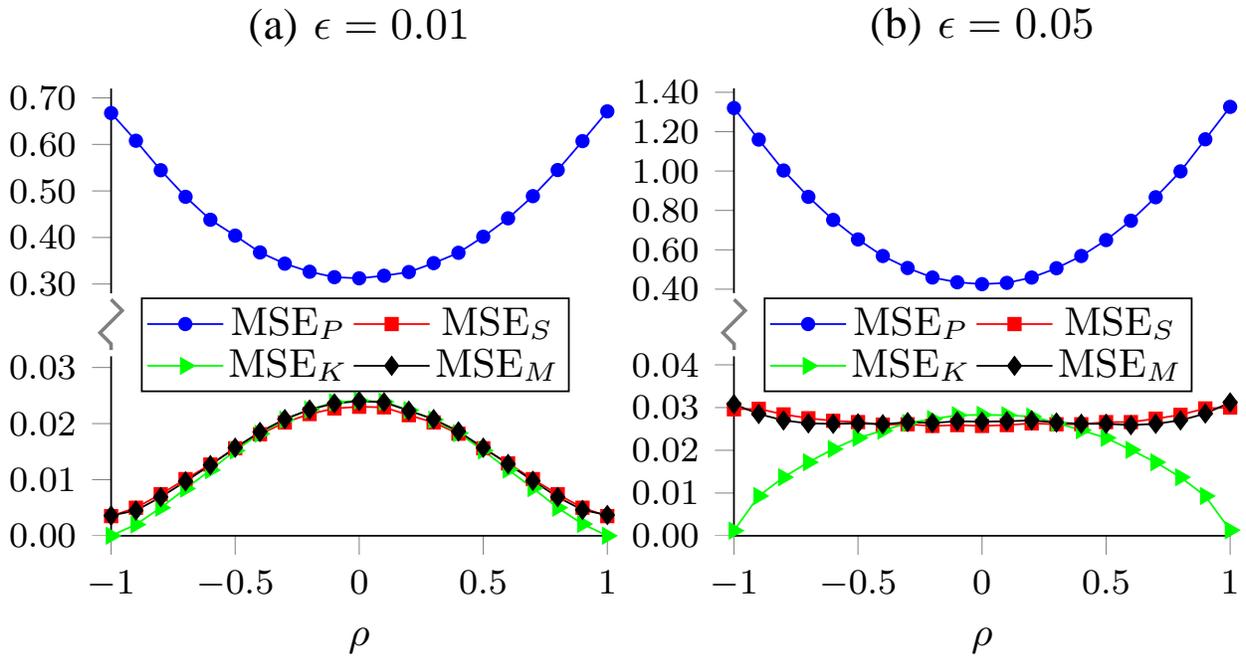}
	\caption{Performance comparison it terms of $\mse_\zeta$, $\zeta\in\{S,P,K,M\}$, over $\rho=-1(0.1)1$ in the 
	contaminated normal model (\ref{eq:gmm}) for (a) $\epsilon=0.01$ and (b) $\epsilon=0.05$ under the sample size $n=50$.
	The rest parameters of the model are set to be $\sigma_X=\sigma_Y=1$, $\lambda_X=\lambda_Y=100$ and $\rho'=0$, respectively.}
	\label{fig:figgmmmse}
\end{figure}
\subsection{Performance of $\hat\rho_\zeta$ in Contaminated Normal Model}
Fig.~\ref{fig:figgmmrsrk} puports to
1) verify the two statements concerning $\mean(\RK)$ and $\mean(\RS)$
in Theorem~\ref{thm:thmrobust}  under the contaminated Gaussian model (\ref{eq:gmm}), and
2) compare our formula (\ref{eq:thmgmmmeanrs}) with the result of ($\ast$) that asserted in~\cite{georgy2002}.
Due to the lack of space, we only present the results for $\epsilon=0.01$ and $\epsilon=0.05$
under the sample size $n=50$ here. For simplicity, the rest parameters of the model (\ref{eq:gmm})
are set to be $\sigma_X=\sigma_Y=1$, $\lambda_X=\lambda_Y=100$ and $\rho'=0$ throughout. 
It is seen that the observed values
of $\mean(\RK)$ and $\mean(\RS)$ agree well with the corresponding theoretical results of 
(\ref{eq:thmgmmmeanrk}) and (\ref{eq:thmgmmmeanrs}) established in Theorem~\ref{thm:thmrobust}.
On the other hand, however, the curves with respect to ($\ast$), especially in Fig.~\ref{fig:figgmmrsrk}(b),
deviate obviously from the corresponding observed values.

Fig.~\ref{fig:figgmmmse} illustrates, in terms of MSE,  the sensitivity of $\hat\rho_P$ 
as well as the robustness of $\hat\rho_S$, $\hat\rho_K$ and $\hat\rho_M$  to impulsive noise.
It is shown in Fig.~\ref{fig:figgmmmse} that the MSE of $\hat\rho_P$ is dramatically larger than those of the
other three estimators, irrespective of how small the  fraction $\epsilon$ of impulsive component is. On the other hand,
it is seen that, despite some minor negative (positive) differences for $\rho$ around $0$ ($\pm1$),
$\mse_S$ and $\mse_M$ behave similarly with $\mse_K$ for $\epsilon=0.01$.
Nevertheless, $\mse_S$ and $\mse_M$ are much larger than $\mse_K$ for $\epsilon=0.05$
when $\rho$ falls in the neighborhood of $\pm1$. Combing Fig.~\ref{fig:figgmmmse}(a) and (b), it would 
be reasonable to rank their performance as $\hat\rho_K\ge\hat\rho_S\sim\hat\rho_M\gg\hat\rho_P$ in terms of MSE 
under the contaminated normal model (\ref{eq:gmm}), where the symbol $\sim$ stands for ``is similar to''. 

\section{Concluding Remarks} \label{sect:sectconclusion}
In this paper we have investigated systematically the properties of the
Spearman's rho and Kendall's tau  for samples drawn from bivariate normal  contained normal populations. 
Theoretical derivations along with Monte Carlo simulations reveal that, contrary to the opinion 
of equivalence between SR and KT in some literature, e.g.~\cite{gilpin1993}, they
behave quite differently in terms of \emph{mathematical tractability}, \emph{bias effect}, \emph{mean square error},
\emph{asymptotic relative efficiency} in the normal cases and \emph{robustness} properties in the contaminated normal model.  

As shown in Theorem~\ref{thm:varrs},  SR is mathematically less tractable than KT, in the sense of
the intractable terms $\Omega_1(\rho)$ and $\Omega_2(\rho)$  in the formula of its variance (\ref{eq:thmvarrs}), in contrast 
with the closed form expression of $\var(\RK)$ in (\ref{eq:lemvarrk}).
However, this mathematical inconvenience is, to some extent, offset by Table~\ref{tab:tabO1O2}
provided in this work, especially from the viewpoint of numerical accuracy.
Moreover, as demonstrated in Fig.~\ref{fig:figbias} and Table~\ref{tab:tabvar},
the convergence speed of the asymptotic formulae (\ref{eq:thmbiask}) and (\ref{eq:varrhok}) 
with respect to $\bias_K$ and $\var(\hat\rho_K)$ are less accurate than those 
of $\bias_S$ and $\var(\hat\rho_S)$ due to the high nonlinearity of the calibration (\ref{eq:estk}).
As a consequence, we do not attach too much importance to such mathematical advantage of KT over SR.

Now let us turn back to the question raised at the very beginning of this paper: which one, SR or KT,
should we use in practice when PPMCC is inapplicable? The answer to this question is different
for different requirements of the task at hand.  Specifically, 
\begin{enumerate}
	\item If \emph{unbiasedness} is on the top priority list, then neither $\hat\rho_S$ or $\hat\rho_K$ should be resorted to.  
	The modified version $\hat\rho_M$ that employs both SR and KT, is definitely the best choice (cf. Fig.~\ref{fig:figbias}).
	\item One the other hand, if minimal MSE is the critical feature and the sample size $n$ is small, then $\hat\rho_S$ ($\hat\rho_K$) 
	     should be employed when the population correlation $\rho$ is weak (strong) (cf. Fig.~\ref{fig:figmse}). 
	\item Since $\hat\rho_K$ outperforms $\hat\rho_S$ asymptotically in terms of  ARE, then $\hat\rho_K$ is the suitable statistic 
	 in large-sample cases (cf. Fig.~\ref{fig:figare}).
	\item If their is impulsive noise in the data, then it would be better to employ $\hat\rho_K$, in terms of MSE, although there is 
		some minor advantage of $\hat\rho_S$ when $\rho$ is in the neighborhood of $0$ (cf. Fig.~\ref{fig:figgmmmse}).
	\item Moreover, in terms of time complexity, $\hat\rho_S$ appears to be superior to $\hat\rho_K$---%
the computational load of the former is $O(n\log n)$; whereas and the computational load of the latter is $O(n^2)$~\cite{xu2007}.
\end{enumerate}

Possessing the desirable properties summarized in Section~\ref{sect:sectgeneral}, Spearman's rho
and Kendall's tau have found wide applications in the literature other than information theory. 
With the new insights uncovered in this paper, these two rank based coefficients
can play complementary roles under the circumstances where Pearson's product moment correlation coefficient
is no longer effective.

\appendices \label{sect:proofs}
\section{Proof of Theorem~\ref{thm:varrs}} \label{app:appvarrs}
\begin{table*}
	\centering
	\caption{Quantities for evaluation of $\mean(\mathcal{S}^2)$ in Theorem \ref{thm:varrs}}
	\includegraphics{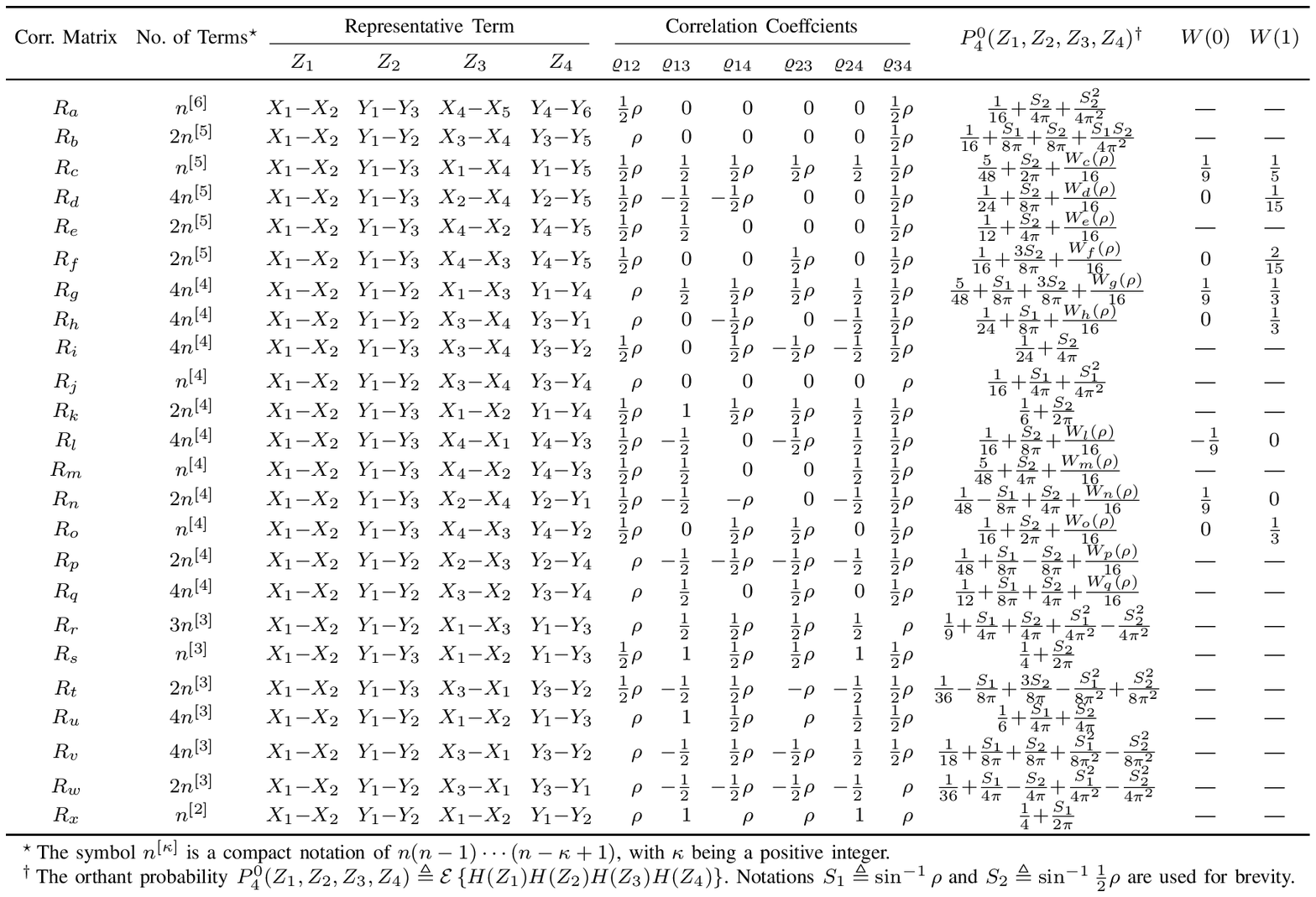}
\label{tab:orthant}
\end{table*}
\begin{proof}
Using the technique developed by Moran~\cite{moran1948} for finding $\mean(\RS)$, it follows that the ranks
can be expressed as
\begin{align}
	\label{eq:Pi}
	P_i &= \sum_{j=1}^n H(X_i-X_j) +1\\
	\label{eq:Qi}
	Q_i &= \sum_{k=1}^n H(Y_i-Y_k) +1
\end{align}
where $H(\blacktriangle)$ is defined in (\ref{eq:Hdef}).
Substituting (\ref{eq:Pi}) and (\ref{eq:Qi}) into (\ref{eq:rsdef}) yields
\begin{equation}
	\RS = \frac{\mathcal{S}-\frac{1}{4}(n-1)^2}{n}\frac{12}{n^2-1}
	\label{eq:sandrs}
\end{equation}
where
\begin{equation}
\label{eq:s}
\begin{split}
	\mathcal{S} &= \sum_{i=1}^n \sum_{j=1}^n \sum_{k=1}^n H(X_i-X_j) H(Y_i-Y_k)\\
	& = \underset{i\ne j=1}{\sum\limits^n\sum\limits^n}H(X_i-X_j)H(Y_i-Y_j) \\
	&\hspace{1cm} + \underset{i\ne j \ne k =1}{\sum\limits^n\sum\limits^n\sum\limits^n}H(X_i-X_j)H(Y_i-Y_k).
\end{split}
\end{equation}
Then
\begin{equation}
	\var(\RS)=\frac{144}{n^2(n^2-1)^2}\underbrace{\bigg[\mean(\mathcal{S}^2)-\mean^2(\mathcal{S})\bigg]}_{\var(\mathcal{S})}.
	\label{eq:varsandrs}
\end{equation}
Taking the expectation of both sides of (\ref{eq:s}) with the assistance of (\ref{eq:p2}) in Lemma~\ref{lem:lemknown}, it follows readily that
\begin{equation}
	\mean(\mathcal{S})=n^{[2]}\left(\frac{1}{4}+\frac{S_1}{2\pi}\right)+n^{[3]}\left(\frac{1}{4}+\frac{S_2}{2\pi}\right)
	\label{eq:means}
\end{equation}
where $n^{[\kappa]}\define n(n{-}1)\cdots(n{-}\kappa{+}1)$, with $\kappa$ being a positive integer. 
Now the variance of $\RS$ depends on the evaluation of $\mean(\mathcal{S}^2)$,  
which is a weighted summation of $24$ quadrivariate normal orthant probabilities $P_4^0(R_\xi)=\mean(Z_1 Z_2 Z_3 Z_4)$ 
corresponding to $R_\xi$ listed in Table~\ref{tab:orthant}~\cite{david1961}. 
Collecting the terms of $P_4^0(R_\xi)$ in Table~\ref{tab:orthant}, subtracting the square of the right side of (\ref{eq:means}) and substituting the resultant into (\ref{eq:varsandrs}) along with some simplifications,
we obtain the expression of (\ref{eq:thmvarrs}) with
\begin{align}
	\label{eq:thmO1second}
	\Omega_1(\rho) & {=} W_c {+} 4W_d {+} 2W_e {+} 2W_f\\ 
	\label{eq:thmO2second}
	\Omega_2(\rho)& {=} 4(W_g {+} W_h {+} W_l {+} W_q){+}2(W_n{+}W_p){+}W_m{+}W_o.
\end{align}
An application of the relationship (\ref{eq:p4}) to Appendix 2 of~\cite{david1961} yields
\begin{equation}
	\label{eq:Widentities}
	W_e {=} 2W_d,\, W_g {=} W_p,\, W_h {=} W_q,\text{ and }W_m {=} 2W_l+\frac{1}{3}.
\end{equation}
Substituting (\ref{eq:Widentities}) into (\ref{eq:thmO1second}) and (\ref{eq:thmO2second}) yields
(\ref{eq:thmO1}) and (\ref{eq:thmO2}), respectively. Hence the first theorem statement (\ref{eq:thmvarrs}) follows.
Ignoring the $o(n^{-1})$ terms in (\ref{eq:thmvarrs}) yields the second statement (\ref{eq:thmvarrsasymp}), thus completing the proof.
\end{proof}
\section{Derivations of $\Omega_1(\rho)$, $\Omega_2(\rho)$ and $\Omega_3(\rho)$ for $\rho=0,1$} \label{app:appO1O2O3}
\begin{proof}
	From (\ref{eq:thmO1}), (\ref{eq:thmO2}) and (\ref{eq:thmcovstatement}),  it suffices to evaluate
	$W_{\xi'}$, $\xi'\in\{c,d,f,g,h,l,n,o\}$ for $\rho=0,1$; and with (\ref{eq:WandP4}),
	it suffices to evaluate $P_4^0(R_{\xi'})$ for $\rho=0,1$. It follows readily from Appendix 2 of~\cite{david1961}
	that for $\rho=0$, $P_4^0(R_c)=P_4^0(R_g)=1/9$, $P_4^0(R_d)=P_4^0(R_h)=1/24$, $P_4^0(R_f)=P_4^0(R_o)=1/16$,
	$P_4^0(R_l)=1/18$, $P_4^0(R_n)=1/36$. Then, with the help of (\ref{eq:WandP4}), we have the values $W_{\xi'}(0)$
	as listed in the $W(0)$-column of Table~\ref{tab:orthant}. Using these $W_{\xi'}(0)$ values with 
	the relationships (\ref{eq:thmO1}), (\ref{eq:thmO2}) and (\ref{eq:thmcovstatement}) yields
        $\Omega_1(0)=1/9$, $\Omega_2(0)=5/9$ and $\Omega_3(0)=1/18$, respectively.

	When $\rho$ approaches unity, it is rather tricky to evaluate the values $W_{\xi'}(1)$. Substituting $\rho=1$ 
	directly into the integrals in (\ref{eq:lemWreduct}) or the integrals in Appendix 2 of~\cite{david1961} 
	will not lead to any tractable argument. We have to investigate case by case. From Table~\ref{tab:orthant}, 
	it is seen that the off-diagonal elements of $R_c$ are all $1/2$ When $\rho=1$. Then we have 
	$P_4^0(R_c)|_{\rho=1}=1/5$~\cite{steck1962}, and hence $W_c(1)=1/5$ by (\ref{eq:WandP4}). From~\cite{plackett1954} it follows
	that $P_4^0(R_f)|_{\rho=1}=2/15$ and $P_4^0(R_m)|_{\rho=1}=1/6$. Then we have, by (\ref{eq:WandP4}), $W_f(1)=2/15$
	and $W_m(1)=1/3$, the latter yielding $W_l(1)=0$ from the identity $W_m=W_l+1/3$ in (\ref{eq:Widentities}).
	Substituting $R_f|_{\rho=1}$ into (\ref{eq:lemWdef}) and exchanging $z_1$ and $z_2$ gives $W_f(1)=W_e(1)$,
	which implies that $W_d(1)=1/15$ by the identity $W_e=2W_d$ in (\ref{eq:Widentities}). Similarly we also have
	$W_o(1)=W_m(1)=1/3$ upon substitution of $R_m|_{\rho=1}$ into (\ref{eq:lemWdef}) and exchange of $z_3$ and $z_4$.
	It is easy to verify that $P_4^0(R_n)$ vanishes as $\rho\to1$, since in this case $Z_1=-Z_4$ and 
	$H(Z_1)H(Z_2)H(Z_3)H(Z_4)\equiv0$ by the definition of H($\blacktriangle$) in (\ref{eq:Hdef}). Then $W_n(1)=0$ 
	by applying the relationship (\ref{eq:WandP4}) once more. When $\rho$ approaches unity, it follows that $P_4^0(R_g)$
	and $P_4^0(R_h)$ degenerate to two trivariate normal orthant probabilities
	that have closed form expressions (\ref{eq:p3}). Specifically, it follows that $P_4(R_g)_{\rho\to1}=1/4$ and 
	$P_4(R_h)|_{\rho\to1}=1/8$, yielding $W_g(1)=1/3$ and $W_h(1)=1/3$, respectively. Having all the values
	of $W_{\xi'}(1)$, as listed in the $W(1)$-column of Table~\ref{tab:orthant}, and the three relationships 
	(\ref{eq:thmO1}), (\ref{eq:thmO2}) and (\ref{eq:thmcovstatement}), we obtain
	$\Omega_1(1)=1$, $\Omega_2(1)=16/3$ and $\Omega_3(1)=1/2$, respectively, and the evaluations complete. 
\end{proof}
\section{Proof of Theorem~\ref{thm:covrsrk}} \label{app:appcovrsrk}
\begin{table*}
	\centering
	\caption{Quantities for evaluation of $\mean(\mathcal{S}\mathcal{T})$ in Theorem \ref{thm:covrsrk}}
\includegraphics{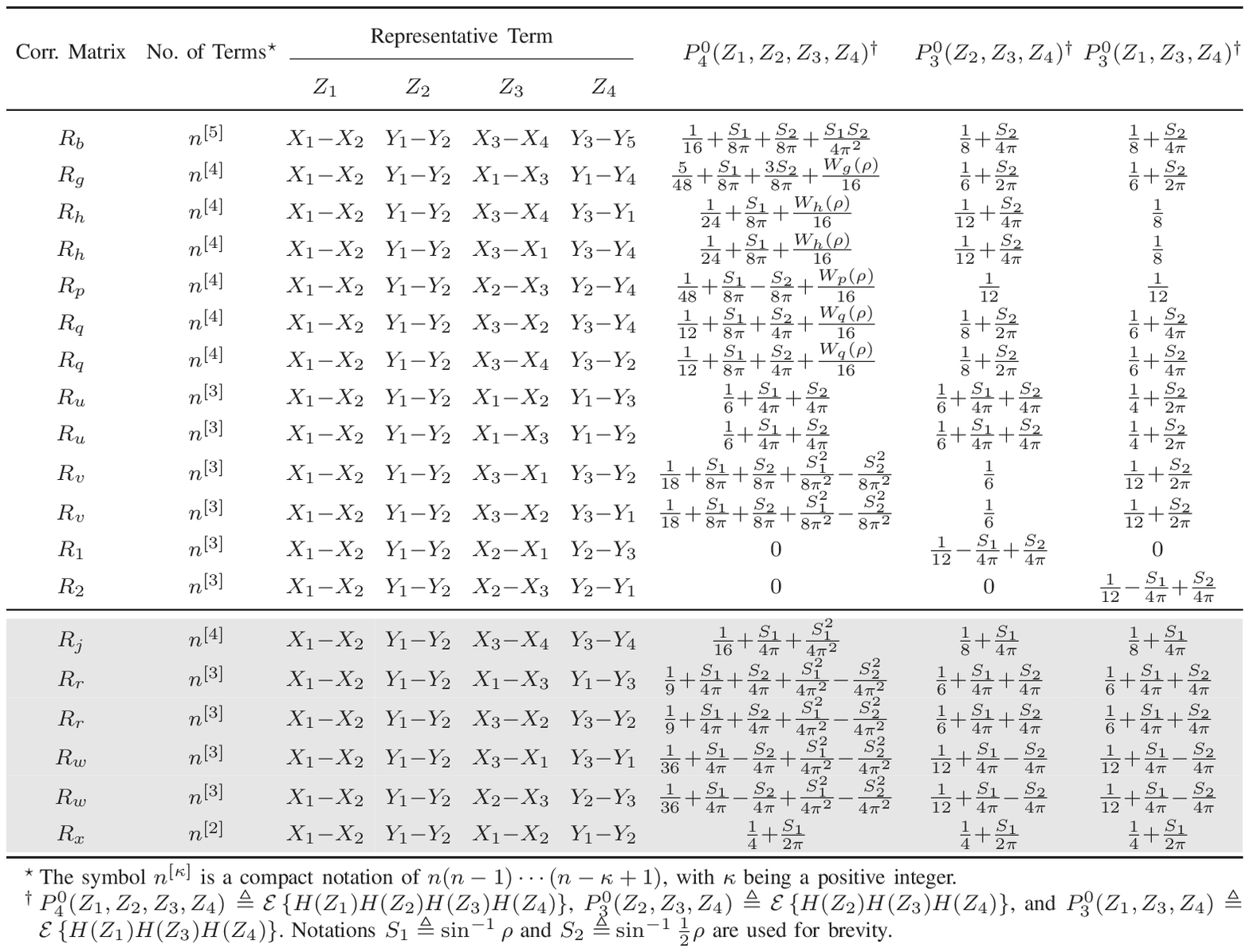}
\label{tab:cov}
\end{table*}
\begin{proof}
	Let $\mathcal{S}$ be the same as in (\ref{eq:s}) and $\mathcal{T}$ be the numerator of (\ref{eq:rkdef}).
	Define
	\begin{align}
		\label{eq:Idef}
		I &\define \underset{i\ne j}{\sum}\,H(X_i-X_j)H(Y_i-Y_j)\\
		\label{eq:Jdef}
		J &\define \hspace{-3pt}\underset{i\ne j\ne k}{\sum}\hspace{-4pt}H(X_i-X_j)H(Y_i-Y_k)\\
		\label{eq:Kdef}
		K &\define \underset{i\ne j}{\sum}\,H(X_i-X_j)\\
		\label{eq:Ldef}
		L &\define \underset{i\ne k}{\sum}\,H(Y_i-Y_k).
	\end{align}
	Then, we have, from (\ref{eq:rkdef}), (\ref{eq:sandrs}) and (\ref{eq:s}) along with the relationship $\sgn(\blacktriangle)=2H(\blacktriangle)-1$,  
	\begin{align}
		\label{eq:ssum}
		\mathcal{S} &= I + J\\
		\label{eq:tsum}
		\mathcal{T} &= 4I -2 K -2L + n^{[2]}
	\end{align}
	and hence
	\begin{equation}
		\begin{split}
			\hspace{-6pt}	\cov(\RS,\RK)  &= \frac{12}{n^2(n-1)(n^2-1)}\cov(\mathcal{S},\mathcal{T})\\
		& = \frac{12}{n^2(n-1)(n^2-1)}\bigg[\mean(\mathcal{S}\mathcal{T})-\mean(\mathcal{S})\mean(\mathcal{T})\bigg].
	\end{split}
		\label{eq:covst}
	\end{equation}
	From (\ref{eq:p1}) and (\ref{eq:p2}),  it follows that 
	\begin{equation*}
		\mean(I)=n^{[2]}\left(\frac{1}{4}+\frac{S_1}{2\pi}\right) \text{ and } \mean(K)=\mean(L)=\frac{n^{[2]}}{2}.
	\end{equation*}
	Substituting these expectation terms into (\ref{eq:tsum}) gives
	\begin{equation}
		\mean(\mathcal{T})=4n^{[2]}\left(\frac{1}{4}+\frac{S_1}{2\pi}\right)-n^{[2]}=\frac{2n^{[2]}}{\pi}S_1.
		\label{eq:meanT}
	\end{equation}
	Recall that we have obtained $\mean(\mathcal{S})$ in (\ref{eq:means}). Now the only difficulty lies in the 
	evaluation of $\mean(\mathcal{S}\mathcal{T})$ in (\ref{eq:covst}). 
	Multiplying (\ref{eq:ssum}) and (\ref{eq:tsum}), expanding and taking expectations term by term, we have 
	\begin{equation}
		\begin{split}
		\mean(\mathcal{S}\mathcal{T})&=4\mean(IJ)-2\mean(KJ)-2\mean(LJ)\\
		&\quad+4\mean(I^2)-2\mean(KI)-2\mean(LI)+n^{[2]}\mean(\mathcal{S}).
	\end{split}
		\label{eq:meanstexp}
	\end{equation}
	Now, resorting to Table~\ref{tab:cov}, we are ready to evaluate the first six terms in (\ref{eq:meanstexp}).  
	From (\ref{eq:Idef}) and (\ref{eq:Jdef}), it follows that $\mean(IJ)$ is a summation of $P_4^0$ terms
	of the form
	\begin{equation}
		\mean\{H(\underbrace{X_i{-}X_j}_{Z_1})H(\underbrace{Y_i{-}Y_j}_{Z_2})
		H(\underbrace{X_k{-}X_l}_{Z_3})H(\underbrace{Y_k{-}Y_m}_{Z_4})\}.
		\label{eq:IJform}
	\end{equation}
	Since, by definition (\ref{eq:Hdef}), $H(0)=0$, the term (\ref{eq:IJform}) vanishes 
	for $i=j$ or $k=l$ or $k=m$. Then there are $n^2(n-1)^2(n-2)$ nontrivial (\ref{eq:IJform})-like terms 
	left to be evaluated. It follows that the domain of the quintuple $(i,j,k,l,m)$ can be partitioned into
	thirteen disjoint and exhaustive subsets whose representative terms, $Z_1$, $Z_2$, $Z_3$, $Z_4$,  are listed in the upper panel of Table~\ref{tab:cov}. 
	Summing up the corresponding $P_4^0$-terms in Table~\ref{tab:cov} leads directly to $\mean(IJ)$.   
	In a similar manner we can obtain $\mean(KJ)$ and $\mean(LJ)$. With the assistance of the lower panel of Table~\ref{tab:cov}, we 
	also have the expressions of $\mean(I^2)$, $\mean(KI)$ and $\mean(LI)$. Substituting these results and 
	(\ref{eq:means}) into (\ref{eq:meanstexp}), subtracting the multiplication of (\ref{eq:means}) and (\ref{eq:meanT})
	and substituting the resultant back into (\ref{eq:covst}), we find that
	$\cov(\RS,\RK)$ is of the form (\ref{eq:thmcovrsrk}) with
	\[
		\Omega_3(\rho) = \frac{1}{4}W_g+\frac{1}{4}W_p+\frac{1}{2}W_h+\frac{1}{2}W_q 
	\]
which simplifies to (\ref{eq:thmcovstatement}) by applying the identities in (\ref{eq:Widentities}).
	The theorem then follows. 
\end{proof}
\section{Proof of Theorem~\ref{thm:thmrobust}} \label{app:appgmm}
\begin{proof}
For ease of the following discussion, we will use $\phi(x,y)$ and $\psi(x,y)$ to denote the pdfs 
of the two bivariate normal components in (\ref{eq:gmm}), respectively. 
From (\ref{eq:Pi}), (\ref{eq:Qi}) and (\ref{eq:tsum}), it follows that the numerator of (\ref{eq:rkdef})
	$\mathcal{T}$ can be simplified to
	\begin{equation}
		\mathcal{T}=4\underset{i\ne j=1}{\sum\limits^n\sum\limits^n}H(X_i-X_j)H(Y_i-Y_j)-n^{[2]}
		\label{eq:Tsumsimp}
	\end{equation}
which yields
\begin{equation}
	\mean(\mathcal{T})=4n^{[2]}\underbrace{\mean\left[H(X_1-X_2)H(Y_1-Y_2)\right]}_{E_1}-n^{[2]}
	\label{eq:meanTsum}
\end{equation}
by the i.i.d. assumption. To evaluate $E_1$ in (\ref{eq:meanTsum}), we need the joint distribution of 
$(X_1, Y_1,X_2,Y_2)$, denoted by $\varphi(x_1,y_1,x_2,y_2)$, which is readily obtained as
\begin{equation}
\begin{split}
	\varphi&=\left[(1-\epsilon)\phi_1+\epsilon\psi_1\right]\left[(1-\epsilon)\phi_2+\epsilon\psi_2\right]\\
	&=\underbrace{(1{-}\epsilon)^2}_{\alpha_1}\underbrace{\phi_1\phi_2}_{\varphi_1}{+}
	\underbrace{\epsilon(1{-}\epsilon)}_{\alpha_2}\underbrace{\phi_1\psi_2}_{\varphi_2}{+}
	\underbrace{\epsilon(1{-}\epsilon)}_{\alpha_3}\underbrace{\phi_2\psi_1}_{\varphi_3}{+}
	\underbrace{\epsilon^2}_{\alpha_4}\underbrace{\psi_1\psi_2}_{\varphi_4}
\end{split}
	\label{eq:pdf1212}
\end{equation}
where $\varphi$, $\phi_i$, $\psi_i$ are compact notations of $\varphi(x_1,y_1,x_2,y_2)$, $\phi(x_i,y_i)$ and $\psi(x_i,y_i)$, $i=1,2$, respectively.
Write 
\[
U\define \frac{X_1-X_2}{\sqrt{\var(X_1-X_2)}} \text{ and } V\define\frac{Y_1-Y_2}{\sqrt{\var(Y_1-Y2)}}.
\]
Then, with respect to $\varphi_1$, $\varphi_2$, $\varphi_3$, and $\varphi_4$ in (\ref{eq:pdf1212}), $(U,V)$ follows four
standard bivariate normal distributions with correlations 
\begin{align}
	\label{eq:rho1}
\varrho_1&=\rho\\
	\label{eq:rho2}
	\varrho_2&=\frac{\rho+\lambda_X\lambda_Y\rho'}{\sqrt{1+\lambda_X^2}\sqrt{1+\lambda_Y^2}}\to\rho'\text{ as }\lambda_X,\lambda_Y\to\infty\\
	\label{eq:rho3}
\varrho_3&=\frac{\rho+\lambda_X\lambda_Y\rho'}{\sqrt{1+\lambda_X^2}\sqrt{1+\lambda_Y^2}}\to\rho'\text{ as }\lambda_X,\lambda_Y\to\infty\\
	\label{eq:rho4}
\varrho_4&=\rho'
\end{align}
respectively. An application of the Sheppard's theorem (\ref{eq:p2}) to (\ref{eq:meanTsum}) along with (\ref{eq:rho1})--(\ref{eq:rho4}) yields
\begin{equation}
\begin{split}
	\mean(\mathcal{T})&=4n^{[2]}\sum_{i=1}^4 \alpha_i\left(\frac{1}{4}+\frac{1}{2\pi}\sin^{-1}\varrho_i\right)-n^{[2]}\\
	&=\frac{2n^{[2]}}{\pi}\left[\alpha_1\sin^{-1}\rho{+}2\alpha_2\sin^{-1}\varrho_2{+}\alpha_4\sin^{-1}\rho'\right].
\end{split}
	\label{eq:meanTsum1}
\end{equation}
Now it is not difficult to verify that the first statement  (\ref{eq:thmgmmmeanrk}) holds
by 1) dividing both sides of (\ref{eq:meanTsum1}) by $n^{[2]}$, 2) letting $\lambda_X\to\infty$ and $\lambda_Y\to\infty$, and 
3) ignoring the $O(\epsilon^2)$ terms.

To prove the second statement (\ref{eq:thmgmmmeanrs}), it suffices to evaluate $\mean(\mathcal{S})$ by the relationship (\ref{eq:sandrs}).
Taking expectations of both sides in (\ref{eq:s}) along with the i.i.d. assumptions gives
\begin{equation}
	\mean(\mathcal{S})=n^{[2]}E_1+n^{[3]}\underbrace{\mean\left[H(X_1-X_2)H(Y_1-Y_3)\right]}_{E_2}.
	\label{eq:meanSsum}
\end{equation}
Since we have known $E_1$ in the above development, now we only need to work out $E_2$ in (\ref{eq:meanSsum}). Let
$\varpi(x_1,y_1,x_2,y_2,x_3,y_3)$, abbreviated as $\varpi$, denote the pdf of the joint distribution of $(X_1,Y_1,X_2,Y_2,X_3,Y_3)$. Then, from (\ref{eq:gmm}) and the i.i.d. assumption,
\begin{equation}
\begin{split}
	\varpi&=\left[(1{-}\epsilon)\phi_1{+}\epsilon\psi_1\right]\left[(1{-}\epsilon)\phi_2{+}\epsilon\psi_2\right]\left[(1{-}\epsilon)\phi_3{+}\epsilon\psi_3\right]\\
	&=(1{-}\epsilon)^3\underbrace{\phi_1\phi_2\phi_3}_{\varpi_1}{+}\epsilon(1{-}\epsilon)^2(\underbrace{\phi_1\phi_2\psi_3}_{\varpi_2}{+}\underbrace{\phi_1\psi_2\phi_3}_{\varpi_3}{+}\underbrace{\psi_1\phi_2\phi_3}_{\varpi_4})\\
	&\phantom{=}+\epsilon^2(1{-}\epsilon)(\underbrace{\phi_1\psi_2\psi_3}_{\varpi_5}{+}\underbrace{\psi_1\phi_2\psi_3}_{\varpi_6}{+}\underbrace{\psi_1\psi_2\phi_3}_{\varpi_7})+\epsilon^3\underbrace{\psi_1\psi_2\psi_3}_{\varpi_8}.
\end{split}
	\label{eq:varpi1to8}
\end{equation}
where $\phi_i$ and $\psi_i$ are compact notations of $\phi(x_i,y_i)$ and $\psi(x_i,y_i)$, $i=1,2,3$, respectively.
Define 
\[
V'=\frac{Y_1-Y_3}{\sqrt{\var(Y_1-Y_3)}}.
\]
Then, with respect to $\varpi_1$ to $\varpi_8$ in (\ref{eq:varpi1to8}), $(U,V')$ follows $8$ standard bivariate normal distributions with correlations 
\begin{align}
	\label{eq:rho5}
	\varrho_5&=\frac{\rho}{2}\\
	\label{eq:rho6}
	\varrho_6&=\frac{1}{\sqrt{2}}\frac{\rho}{\sqrt{1+\lambda_Y^2}}\to 0\text{ as } \lambda_Y\to\infty\\
	\label{eq:rho7}
	\varrho_7&=\frac{1}{\sqrt{2}}\frac{\rho}{\sqrt{1+\lambda_X^2}}\to 0\text{ as } \lambda_X\to\infty\\
	\label{eq:rho8}
	\varrho_8&=\frac{\lambda_X\lambda_Y\rho'}{\sqrt{1+\lambda_X^2}\sqrt{1+\lambda_Y^2}}\to \rho'\text{ as } \lambda_X,\lambda_Y\to\infty \\
	\label{eq:rho9}
	\varrho_9&=\frac{\rho}{\sqrt{1+\lambda_X^2}\sqrt{1+\lambda_Y^2}}\to 0\text{ as } \lambda_X,\lambda_Y\to\infty\\
	\label{eq:rho10}
	\varrho_{10}&=\frac{1}{\sqrt{2}}\frac{\lambda_X\rho'}{\sqrt{1+\lambda_X^2}}\to \frac{\rho'}{\sqrt{2}}\text{ as } \lambda_X\to\infty\\
	\label{eq:rho11}
	\varrho_{11}&=\frac{1}{\sqrt{2}}\frac{\lambda_Y\rho'}{\sqrt{1+\lambda_Y^2}}\to \frac{\rho'}{\sqrt{2}}\text{ as } \lambda_Y\to\infty\\
	\label{eq:rho12}
	\varrho_{12}&=\frac{\rho'}{2}.
\end{align}
Using the Sheppard's theorem (\ref{eq:p2}) again together with (\ref{eq:varpi1to8})--(\ref{eq:rho12}), we can obtain the
expression of $E_2$ and hence $\mean(\mathcal{S})$ in terms of $n$, $\epsilon$ and $\varrho_1$ to $\varrho_{12}$. Substituting $\mean(\mathcal{S})$ into (\ref{eq:sandrs}),
letting $n$, $\lambda_X$, $\lambda_Y\to\infty$ and ignoring the $O(\epsilon^2)$ terms, we arrive at (\ref{eq:thmgmmmeanrs}), the second theorem statement. 
\end{proof}

\begin{biographynophoto}{Weichao Xu}
(M'06) received the B.Eng. and M.Eng. degrees in electrical
engineering from the University of Science and Technology of China, Hefei, China, in 1993
and 1996, respectively. He received the Ph.D. degree in biomedical
engineering from the University of Hong Kong, Hong Kong, in 2002.
Since 2003, he has been a Research Associate with the Department of Electrical and Electronic Engineering, the
University of Hong Kong. His research interests are in the areas of
mathematical statistics, machine learning, digital signal processing and applications.
\end{biographynophoto}
\begin{biographynophoto}{Yunhe Hou}
Yunhe Hou (M'06) received the B.E (1999), M.E(2002) and Ph.D(2005)
degrees from the Huazhong University of Science and Technology, China.
He worked as a postdoctoral research fellow at Tsinghua University
from 2005 to 2007. He was a visiting scholar at Iowa State University,
Ames, and a researcher of University College Dublin, Ireland from 2008
to 2009. He is currently with the University of Hong Kong, Hong Kong,
as a research assistant professor.
\end{biographynophoto}
\begin{biographynophoto}{Y. S. Hung}
(M'88--SM'02) received the B.Sc. (Eng.)
degree in electrical engineering and the B.Sc. degree
in mathematics from the University of Hong Kong,
Hong Kong, and the M.Phil. and Ph.D. degrees from
the University of Cambridge, Cambridge, U.K.

He was a Research Associate with the University
of Cambridge and a Lecturer with the University
of Surrey, Surrey, U.K. In 1989, he joined the
University of Hong Kong, where he is currently a
Professor.  His research interests include robust control systems
theory, robotics, computer vision, and biomedical engineering.

Prof. Hung was a recipient of the Best Teaching Award in 1991 from the Hong
Kong University Students¡¯ Union. He is a chartered engineer and a fellow of
IET and HKIE.
\end{biographynophoto}
\begin{biographynophoto}{Yuexian Zou}
Yuexian Zou received her M.Sc. (1991) and Ph.D. (2000) from the University of Electronic Science and
Technology of China and the University of Hong Kong, respectively.

Since 2006, she serves as an Associate Professor in Peking University, and is the director
of the Advanced Digital Signal Processing Lab of Peking University Shenzhen Graduate School. 
Dr. Zou has more than 15 years research experience in digital signal processing, 
adaptive signal processing and their applications. She is currently the senior IEEE member
and has published more than 50 journal and conference papers.

She was the organization co-chair of NEMS09 and served as the founding chair of the WIE Singapore in 2005.
She has carried out more than 10 national funded projects since 2000.
She is currently working on a project funded by Shenzhen Bureau of Science Technology and Information in fast extraction
of Somatosensory Evoked Potential. Dr. Zou¡¯s research interests
include adaptive signal processing, biomedical signal processing and active noise control. 
\end{biographynophoto}
\end{document}